\documentclass[sigplan,screen]{acmart}
\pdfoutput=1

\setcopyright{acmlicensed}
\acmPrice{15.00}
\acmDOI{10.1145/3453483.3454110}
\acmYear{2021}
\copyrightyear{2021}
\acmSubmissionID{pldi21main-p826-p}
\acmISBN{978-1-4503-8391-2/21/06}
\acmConference[PLDI '21]{Proceedings of the 42nd ACM SIGPLAN International Conference on Programming Language Design and Implementation}{June 20--25, 2021}{Virtual, Canada}
\acmBooktitle{Proceedings of the 42nd ACM SIGPLAN International Conference on Programming Language Design and Implementation (PLDI '21), June 20--25, 2021, Virtual, Canada}

\usepackage{amsthm,amsmath,stmaryrd}
\usepackage{thmtools}
\usepackage{thm-restate}
\usepackage{subcaption}

\usepackage[utf8]{inputenc}
\usepackage[T1]{fontenc}
\usepackage{microtype}
\usepackage{tikz}
\usepackage{pgfplots}
\usetikzlibrary{calc,matrix,arrows,automata,positioning,chains}

\newcommand{\tikzsavept}[1]{
  \tikz[baseline,remember picture]{\node[anchor=base,minimum height=1ex] (#1) {};}
}

\usepackage{tikz-cd}
\pgfplotsset{compat=1.14}

\newcommand{\tuple}[1]{\left\langle #1 \right\rangle}     % n-tuples
\newcommand{\sem}[1]{\left\llbracket #1 \right\rrbracket} % Semantic brackets
\newcommand{\set}[1]{\left\{ #1 \right\}}                 % Sets
\newcommand{\true}{\textit{true}}
\newcommand{\false}{\textit{false}}
\newcommand{\defeq}{\triangleq}
\renewcommand{\vec}[1]{\mathbf{#1}}

\definecolor{black}{RGB}{0,0,0}
\definecolor{green}{RGB}{25,160,80}
\definecolor{red}{RGB}{192,57,43}
\definecolor{blue}{RGB}{41,128,185}
\definecolor{orange}{RGB}{255,99,0}
\definecolor{gray}{RGB}{75,75,75}
\definecolor{lightgrey}{RGB}{240,240,240}
\definecolor{purple}{RGB}{155, 89, 182}
\definecolor{darkgrey}{RGB}{52, 73, 94}

\usepackage[noend,linesnumbered]{algorithm2e}
\SetKwInOut{Input}{Input}
\SetKwInOut{Output}{Output}
\SetKwRepeat{Do}{do}{while}

\SetCommentSty{mycommfont}
\SetKwProg{Fn}{Subroutine}{ begin}{end}
\SetKwProg{Alg}{Algorithm}{ begin}{end}
\SetKwComment{Inv}{\rm\textbf{invariant:} }{}

\usepackage{listings}
\lstdefinestyle{base}{
  language=C,
  emptylines=1,
  breaklines=true,
  basicstyle=\color{black}\fontfamily{pzc}\selectfont\tt,
  commentstyle=\color{gray}\rm\itshape,
  keywordstyle=\rm\bfseries,
  identifierstyle=\rm\itshape,
  escapechar=@,
  morekeywords=[1]{then,do,halt},
  morekeywords=[2]{assert},
  morekeywords=[3]{requires,ensures},
  keywordstyle	= [2]\color{red}\bf,
  keywordstyle	= [3]\bf\color{gray},
  numberstyle=\footnotesize\color{gray},
  moredelim=**[is][\bf\color{green}]{@!}{!@},
  moredelim=**[is][\bf\color{red}]{@?}{?@},
  literate=
    {<=}{{$\leq$}}1
    {>=}{{$\geq$}}1
    {!}{{$\neg$}}1
    {!=}{{$\neq$}}1
    {||}{{$\lor$}}1
    {&&}{{$\land$}}1
    {->}{{$\rightarrow$}}1
    {_1}{$_{1}$}2
    {_2}{$_{2}$}2
    {_3}{$_{3}$}2
}
\newcommand{\cinline}[1]{\mbox{\lstinline[style=base,mathescape]{#1}}}

\usepackage{mathpartir} % Inference rules
\usepackage{booktabs}   % Tables
\usepackage{hyperref}   % hyperlinks
\usepackage{cleveref}   % better cross-references
\usepackage{thmtools} 
\usepackage{thm-restate}
\usepackage{multirow}
\usepackage{adjustbox}

\DeclareFontFamily{U}  {MnSymbolC}{}
\DeclareFontShape{U}{MnSymbolC}{m}{n}{
    <-6>  MnSymbolC5
   <6-7>  MnSymbolC6
   <7-8>  MnSymbolC7
   <8-9>  MnSymbolC8
   <9-10> MnSymbolC9
  <10-12> MnSymbolC10
  <12->   MnSymbolC12}{}
\DeclareSymbolFont{MnSyC}{U}{MnSymbolC}{m}{n}
\DeclareFontFamily{U}  {MnSymbolA}{}
\DeclareFontShape{U}{MnSymbolA}{m}{n}{
    <-6>  MnSymbolA5
   <6-7>  MnSymbolA6
   <7-8>  MnSymbolA7
   <8-9>  MnSymbolA8
   <9-10> MnSymbolA9
  <10-12> MnSymbolA10
  <12->   MnSymbolA12}{}
\DeclareSymbolFont{MnSyA}{U}{MnSymbolA}{m}{n}
\DeclareMathSymbol{\righthalfcup}{\mathrel}{MnSyC}{184}
\DeclareMathSymbol{\leftmodels}{\mathrel}{MnSyA}{226}
\newcommand{\romanqed}{$\righthalfcup$}

\newenvironment{example}[1][]{%
  \par%
  \pushQED{\let\qedsymbol\romanqed\qed}%
  \trivlist%
\item\ignorespaces%
  \refstepcounter{theorem}%
  \textbf{Example \arabic{section}.\arabic{theorem}}\ifx&#1& \else \;(#1) \fi \hspace{2pt}
}{%
  \popQED\endtrivlist%
}

\usepackage{mdframed}

\newmdenv[
    linecolor=red,
    topline=false,
    bottomline=false,
    rightline=false,
    outermargin=-11pt,
    innermargin=-11pt,
    linewidth=2pt
]{changebar}

\newcommand{\interp}{\mathcal{I}}
\newcommand{\elbl}{L}

\renewcommand{\mp}{\textit{mp}}
\renewcommand{\wp}{\textit{wp}}
\newcommand{\TF}{\textbf{TF}}
\newcommand{\SF}{\textbf{SF}}
\newcommand{\MP}{\textbf{MP}}

\newcommand{\conv}{\textit{conv}}
\newcommand{\Var}{\textsf{Var}}
\newcommand{\LVar}{\textsf{LVar}}
\newcommand{\GVar}{\textsf{GVar}}
\newcommand{\State}{\textsf{State}}
\newcommand{\RegExp}{\textsf{RegExp}}
\newcommand{\oRegExp}{{\omega\textsf{-RegExp}}}

\newcommand{\tfsem}[1]{\mathcal{T}\!\sem{#1}}

\newcommand{\tr}[3]{#2 \rightarrow_{#1} #3}
\newcommand{\trstar}[3]{#2 \rightarrow_{#1}^* #3}
\newcommand{\mpsem}[1]{\mathcal{T}^\omega\!\sem{#1}}
\newcommand{\abssem}[1]{\interp\!\sem{#1}}
\newcommand{\absosem}[1]{\interp^\omega\!\!\sem{#1}}

\newcommand{\pathexp}[3]{\textit{PathExp}_{#1}\!\left(#2,#3\right)}
\newcommand{\opathexp}[2]{\textit{PathExp}_{#1}^\omega\!\left(#2\right)}
\newcommand{\paths}[3]{\textit{Paths}_{#1}\!\left(#2,#3\right)}
\newcommand{\opaths}[2]{\textit{Paths}_{#1}^\omega\!\left(#2\right)}
\newcommand{\pathrep}[3]{\textit{PathRep}_{#1}\!\left(#2,#3\right)}
\newcommand{\opathrep}[2]{\textit{PathRep}_{#1}^\omega\!\left(#2\right)}
\newcommand{\feval}{\textsf{eval}}
\newcommand{\flink}{\textsf{link}}
\newcommand{\ffind}{\textsf{find}}
\newcommand{\tightparen}[1]{\left(\mkern-5mu #1 \mkern-5mu\right)}
\newcommand{\summary}{S}

\begin{document}

\title{Termination Analysis without the Tears}

\author{Shaowei Zhu}
\email{shaoweiz@cs.princeton.edu}
\affiliation{%
	\institution{Princeton University}
	\city{Princeton}
	\state{NJ}
	\country{USA}
}

\author{Zachary Kincaid}
\email{zkincaid@cs.princeton.edu}
\affiliation{%
  \institution{Princeton University}
  \city{Princeton}
  \state{NJ}
  \country{USA}
}

\begin{abstract}
  Determining whether a given program terminates is the quintessential
undecidable problem.  Algorithms for termination analysis may be classified
into two groups: (1) algorithms with strong behavioral guarantees that
work in limited circumstances (e.g., complete synthesis of linear
ranking functions for polyhedral loops), and (2) algorithms that are
widely applicable, but have weak behavioral guarantees
(e.g., {\sc Terminator}).  This paper investigates the space in
between: \textit{how can we design practical termination analyzers with
  useful behavioral guarantees?}

This paper presents a termination analysis that is both
\textit{compositional} (the result of analyzing a composite program is
a function of the analysis results of its components) and
\textit{monotone} (``more information into the analysis yields more
information out'').  The paper has two key contributions.  The first
is an extension of Tarjan's method for solving path problems in graphs
to solve \textit{infinite} path problems.  This provides a foundation
upon which to build compositional termination analyses.  The second is
a collection of monotone conditional termination analyses based on
this framework.  We demonstrate that our tool ComPACT
(Compositional and Predictable Analysis for Conditional Termination)
is competitive with state-of-the-art termination tools while providing
stronger behavioral guarantees.

\end{abstract}




\maketitle

\section{Introduction} \label{sec:introduction}

Termination is an important correctness property in itself, and is a
sub-problem of proving total correctness, liveness properties
\cite{POPL:CGPRV2007,POPL:CK2011,PLDI:CK2013,CAV:CKP2015,TACAS:BCIKP2016},
and bounds on resource usage
\cite{SAS:AAGP2008,POPL:GMC2009,PLDI:GZ2010,PLDI:CHS2015,FMCAD:SZV2015}.
Determining whether a program terminates is undecidable, and so
progress on automated tools for termination analysis is driven by
heuristic reasoning techniques.  While these heuristics are often
effective in practice, they can be brittle and unpredictable.  For
example, termination analyzers may themselves fail to terminate on
some input programs, or report false alarms, or return different
results for the same input, or suffer from ``butterfly effects'', in
which a small changes to the program's source code drastically changes
the analysis.

%% Such effects are particularly problematic for
%% user-in-the-loop verification efforts.  As \citet{UVW:LM2010} put it
%% succinctly: ``tools can understand our programs, but we cannot
%% understand our tools.''

This paper is motivated by the principle that \textit{changes to a
  program should have a predictable impact on its analysis}.  We
develop a style of termination analysis that achieves two particular
desiderata:
\begin{itemize}
\item \textit{Compositionality}: composite programs are analyzed by
  analyzing their sub-components and then combining the results.
  Compositionality implies that changing part of a program only
  changes the analysis of that part.  It enables prompt user
  interaction, since an analysis need not reanalyze the whole program
  to respond to a program change.
\item \textit{Monotonicity}: more information into the analysis yields
  more information out.  Monotonicity implies that certain actions,
  e.g., a user annotating a procedure with additional pre-conditions,
  or an abstract interpreter instrumenting loop invariants into a
  program, \textit{cannot} degrade analysis results.
\end{itemize}

Our approach is based on the paradigm of \textit{algebraic program
  analysis} \cite{JACM:Tarjan1981,JACM:Tarjan1981b,FMCAD:FK2015}.
An algebraic program analysis is described by an algebraic structure in which the elements
represent properties of finite program executions and the operations compose
those properties via sequencing, choice, and iteration (mirroring the
structure of regular expressions).  To verify a safety property, an
algebraic program analyzer computes a regular expression recognizing
all paths through a program to a point of interest, interprets the
regular expression within the given structure, and checks whether the
resulting property entails the property of interest.

In this paper, we extend the algebraic approach to reason about
\textit{infinite} program paths, and thereby provide a conceptual and
algorithmic basis for compositional analysis of liveness properties
such as termination.  Conceptually, our method proves that a program
terminates by computing a transition formula for each loop that
over-approximates the behavior of its body, and then proving that the
corresponding relation admits no infinite sequences.  (Our approach
is not unique in this regard (see Section~\ref{sec:related})---we
provide a unifying foundation for such analyses).

A drawback of using summaries to prove termination is that the loop
body summary \textit{over-approximates} its behavior, and so the
summary may not terminate even if the original loop does.  The
advantage is that we can reason about the summary effectively, whereas
any non-trivial question about behavior of the original loop is
undecidable.  This is the key idea that enables the design of
\textit{monotone} termination analyses.

A particular challenge of compositional termination analysis is that
termination arguments must be synthesized independently of the
surrounding context of the loop (that is, without supporting invariants).
We meet this challenge with a set of methods that exploit loop
summarization to generate monotone \textit{conditional} termination
arguments.  These methods synthesize both a termination argument and
an initial condition under which that argument holds, with the latter
acting as a surrogate for a supporting invariant.

\noindent \paragraph{Contributions}
The contributions of this paper are:
\begin{itemize}
\item A framework for designing compositional analyses of infinite
  program paths.  This framework extends Tarjan's method for solving
  path problems \cite{JACM:Tarjan1981,JACM:Tarjan1981b} from finite to
  infinite paths.
\item An efficient algorithm for computing an $\omega$-regular
  expression that recognizes the infinite paths through a control flow
  graph, which forms the algorithmic foundation of our program
  analysis framework.
%\item We show that our framework can be applied to the problem of
%  computing sufficient preconditions for program termination.  We
%  present the first termination analysis that is compositional,
%  monotone, and applies to a general program model.
%% \item A framework for designing compositional conditional termination
%%   analyses.  This framework extends Tarjan's algebraic program
%%   analysis framework \cite{JACM:Tarjan1981,JACM:Tarjan1981b}, which
%%   generates summaries of program behavior by first computing
%%   a regular expression that represents all paths through a program
%%   and then interpreting that regular expression in an algebra of
%%   summaries.  To adapt this style of analysis to prove termination, we
%%   use $\omega$-regular expressions to represent all infinite paths
%%   through the program, and extend the algebra with summaries with an
%%   algebra of \textit{mortal preconditions}, or \textit{sufficient
%%     conditions for termination}.  The essential operation of interest
%%   is the $\omega$ operation, which given a summary representing the
%%   body of a loop, computes a \textit{mortal precondition}.  The drawback of this
%%   approach is that the loop body summary \textit{over-approximates}
%%   its behavior, and so the summary may not terminate even if the
%%   original loop does.  The advantage is that we can reason about the
%%   summary effectively, whereas any non-trivial question about behavior
%%   of the original loop is undecidable; this is the key idea that
%%   enables monotonicity of termination analyses.
\item The first termination analysis that is compositional, monotone, and applies to a general program model.  We present a set of combinators for constructing a family of such (conditional)
  termination analyses based on our framework. In particular, we introduce \textit{phase analysis}, which improves the precision of a given conditional
  termination analysis by partitioning the space of transitions in a
  loop into phases.

%% we describe several ways to construct new conditional
%%   termination analyses whose results are guaranteed to be no worse than the
%%   existing ones. 
%%   and \textit{attractor region
%%   analysis}, which tries to reason about the termination of a loop restricted to
%%   a certain set of initial states and conclude the termination of the loop for
%%   arbitrary initial state.

%% Using this framework, an analysis
%%   designer is responsible for devising a procedure to compute mortal
%%   preconditions for summaries, and the framework lifts this procedure
%%   to compute mortal preconditions for whole programs (including nested
%%   loops and recursive procedures).  
\end{itemize}

\section{Overview} \label{sec:overview}

In Section~\ref{sec:ata}, we define an algebraic framework for
analyzing liveness properties of programs.  An analysis proceeds in
two steps: (1) compute an $\omega$-regular expression that recognizes
the paths through a program, and (2) interpret that $\omega$-regular
expression with an algebraic structure corresponding to a program
analysis of interest.

We illustrate this process in Figure~\ref{fig:overview-ex}.  Consider
the example program given by its control flow graph (CFG) in
Figure~\ref{fig:overview-ex-cfg} (concrete
syntax for the CFG is given in
Figure~\ref{fig:overview-ex-ast}).  Note that conditional control flow
is encoded as \textit{assumptions}, which do
not change the program variables but can only be executed if the
assumed condition holds (e.g., if the program is in a state where
\cinline{m} is less than \cinline{step} then it may execute the
assumption \cinline{[m < step]}, otherwise it is blocked).

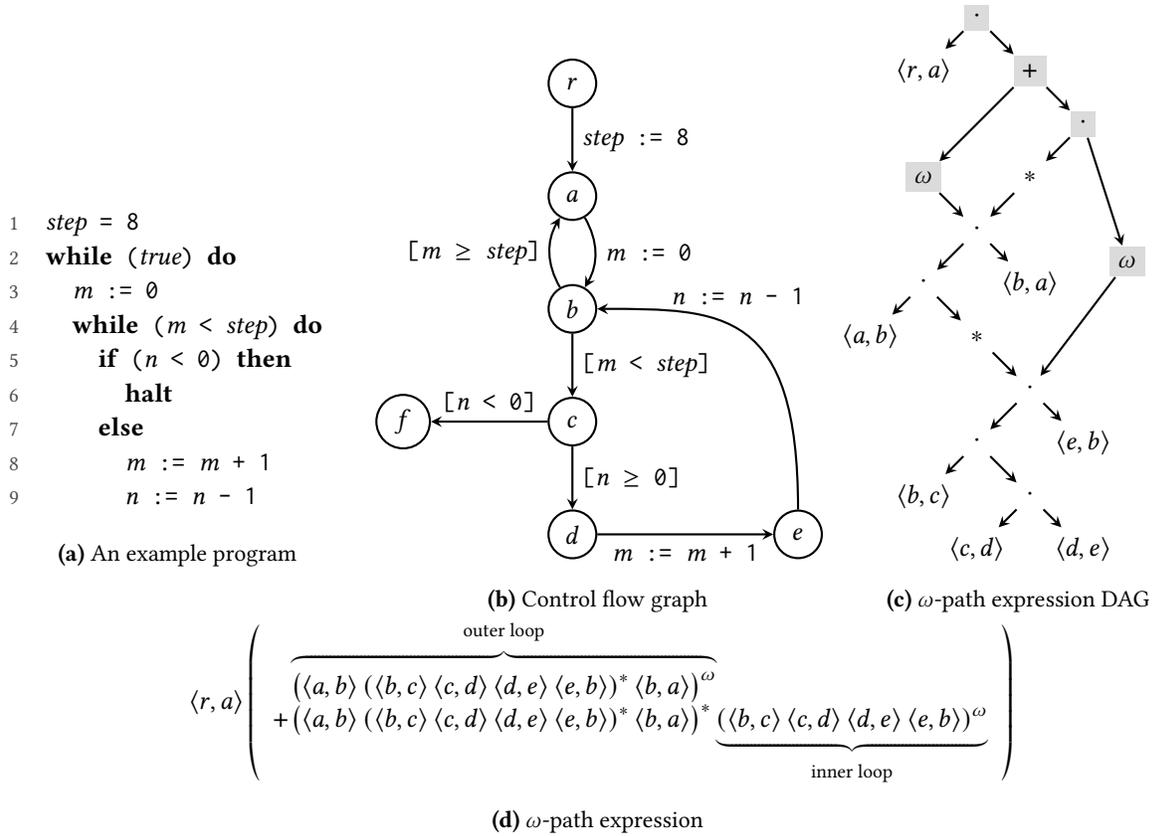
\begin{figure*}
  \begin{subfigure}[b]{5cm}
\begin{lstlisting}[style=base,numbers=left,xleftmargin=2em]
step = 8
while (true) do @\tikzsavept{overview-ex-outer}@
  m := 0
  while (m < step) do @\tikzsavept{overview-ex-inner}@
    if (n < 0) then
      halt
    else
      m := m + 1
      n := n - 1 @\tikzsavept{overview-ex-end}@
\end{lstlisting}
\caption{An example program \label{fig:overview-ex-ast}}
  \end{subfigure}
  \begin{subfigure}[t]{6cm}
\begin{tikzpicture}[thick,>=stealth,
      v/.style={circle,draw,minimum height=18pt},
      node distance=1.5cm]

  \node [v] (r) {$r$};
  \node [v,below of=r] (a) {$a$};
  \node [v,below of=a] (b) {$b$};
  \node [v,below of=b] (c) {$c$};
  \node [v,below of=c] (d) {$d$};
  \node [v,right of=d, node distance = 3cm] (e) {$e$};
  \node [v,left of=c,node distance=2.25cm] (f) {$f$};
  \draw (r) edge[->] node[right]{\cinline{step := 8}} (a);
  \draw (a) edge[->,bend left] node[right]{\cinline{m := 0}} (b);
  \draw (b) edge[->,bend left] node[left]{\cinline{[m >= step]}} (a);
  \draw (b) edge[->] node[right]{\cinline{[m < step]}} (c);
  \draw (c) edge[->] node[above]{\cinline{[n < 0]}} (f);
  \draw (c) edge[->] node[right]{\cinline{[n >= 0]}} (d);
  \draw (d) edge[->] node[below]{\cinline{m := m + 1}} (e);
  \draw [->] (e) .. controls +(90:3cm) and +(0:2cm) .. (b);
  \node [right of=b,xshift=20pt,yshift=5pt] {\cinline{n := n - 1}};
%  edge[->,bend right] node[right]{\cinline{n = n - 1}}
\end{tikzpicture}
\caption{Control flow graph \label{fig:overview-ex-cfg}}
  \end{subfigure}
  \begin{subfigure}[t]{5cm}
\begin{tikzpicture}[thick,>=stealth,omega/.style={fill=gray!20},]
  \node [omega] (abcdebao) {$\omega$};
  \node [below right of=abcdebao] (abcdeba) {$\cdot$};
  \node [above right of=abcdeba] (abcdebas) {$*$};
  \node [above right of=abcdebas,omega] (snd) {$\cdot$};
  \node [above left of=snd,omega] (sum) {$+$};
  \node [above left of=sum,omega] (root) {$\cdot$};
  \node [below left of=root] (ra) {$\tuple{r,a}$};
  \node [below right of=abcdeba] (ba) {$\tuple{b,a}$};
  \node [below left of=abcdeba] (abcdeb) {$\cdot$};
  \node [below left of=abcdeb] (ab) {$\tuple{a,b}$};
  \node [below right of=abcdeb] (bcdebs) {$*$};
  \node [right of=bcdebs,xshift=1cm,yshift=1cm,omega] (bcdebo) {$\omega$};
  \node [below right of=bcdebs] (bcdeb) {$\cdot$};
  \node [below right of=bcdeb] (eb) {$\tuple{e,b}$};
  \node [below left of=bcdeb] (bcde) {$\cdot$};
  \node [below left of=bcde] (bc) {$\tuple{b,c}$};
  \node [below right of=bcde] (cde) {$\cdot$};
  \node [below left of=cde] (cd) {$\tuple{c,d}$};
  \node [below right of=cde] (de) {$\tuple{d,e}$};
  \draw (cde) edge[->] (cd);
  \draw (cde) edge[->] (de);
  \draw (bcde) edge[->] (cde);
  \draw (bcde) edge[->] (bc);
  \draw (bcdeb) edge[->] (bcde);
  \draw (bcdeb) edge[->] (eb);
  \draw (bcdebs) edge[->] (bcdeb);
  \draw (abcdeb) edge[->] (bcdebs);
  \draw (abcdeb) edge[->] (ab);
  \draw (abcdeba) edge[->] (abcdeb);
  \draw (abcdeba) edge[->] (ba);
  \draw (abcdebao) edge[->] (abcdeba);
  \draw (abcdebas) edge[->] (abcdeba);
  \draw (bcdebo) edge[->] (bcdeb);
  \draw (snd) edge[->] (abcdebas);
  \draw (snd) edge[->] (bcdebo);
  \draw (abcdebao) edge[<-] (sum);
  \draw (sum) edge[->] (snd);
  \draw (sum) edge[<-] (root);
  \draw (ra) edge[<-] (root);
\end{tikzpicture}
\caption{$\omega$-path expression DAG \label{fig:overview-ex-dag}}
  \end{subfigure}
  \begin{subfigure}{\textwidth}
\[
\tuple{r,a}\left(\begin{array}{l@{\hspace*{1pt}}l}
  & \overbrace{\left(\tuple{a,b}\left(\tuple{b,c}\tuple{c,d}\tuple{d,e}\tuple{e,b}\right)^*\tuple{b,a}\right)^\omega}^{\text{outer loop}}\\
  + & \left(\tuple{a,b}\left(\tuple{b,c}\tuple{c,d}\tuple{d,e}\tuple{e,b}\right)^*\tuple{b,a}\right)^*\underbrace{\left(\tuple{b,c}\tuple{c,d}\tuple{d,e}\tuple{e,b}\right)^\omega}_{\text{inner loop}}
\end{array}\right)
\]
\caption{$\omega$-path expression \label{fig:overview-ex-pe}}
  \end{subfigure}
\caption{An example program, its control flow graph, and a corresponding $\omega$-path expression \label{fig:overview-ex}}
\end{figure*}

\paragraph{Step 1: Compute an $\omega$-path expression}

Using the algorithm described in Section~\ref{sec:pathexp}, we can
compute an $\omega$-regular expression
that represents all infinite paths in the
CFG that begin at the entry vertex $r$
(Figure~\ref{fig:overview-ex-pe}). 
The  expression can
be represented efficiently as a directed acyclic graph (DAG), where each leaf is labeled by a
control flow edge, each internal node with an operator
(one of: choice ($+$), concatenation ($\cdot$), iteration ($*$), or infinite repetition ($\omega$)---see Section~\ref{sec:flow-graphs}),
and edges are
drawn from operators to operands (Figure~\ref{fig:overview-ex-dag}).
Observe that each node in the DAG
corresponds to either a regular expression (white nodes) or an
$\omega$-regular expression (gray nodes).

\paragraph{Step 2: Interpretation}
The result of a particular analysis is computed by interpreting a
path expression for a program within some abstract domain.  The
domain consists of (1) a \textit{regular algebra}, which is equipped
with choice, concatenation, and iteration operators and which can be
used to interpret regular expressions, and (2) an
\textit{$\omega$-regular algebra}, which is equipped with choice,
concatenation, and $\omega$-iteration operators, and which can be used to
interpret $\omega$-regular expressions.

Our main interest in this paper is in a family of termination
analyses.  In this family, the regular algebra is the algebra of
\textit{transition formulas}, which we denote by $\TF$.  A
\textit{transition formula} is a logical formula over the variables of
the program (in Figure~\ref{fig:overview-ex}: $m,n,\cinline{step}$) along with primed copies
($m',n',\cinline{step}'$) representing the program variables before
and after executing a computation, respectively.  The choice operation
for $\TF$ is disjunction, concatenation is relational composition, and
iteration over-approximates reflexive transitive closure (a
particular iteration operator is defined in Section~\ref{sec:background}).
  The $\omega$-regular algebra is an
algebra of \textit{mortal preconditions}, which we denote by $\MP$; in
fact, we will define several such algebras in this paper, but they
share a common structure.  A mortal precondition is a state formula
(over the program variables ($m,n,\cinline{step}$)) that is satisfied only by \textit{mortal states}, from which the program must terminate.
The choice operation for $\MP$ is
conjunction (a mortal state must be mortal on \textit{all} paths),
concatenation is weakest precondition (a state is mortal only if it can
reach only mortal states), and $\omega$-iteration computes a mortal precondition
for a transition formula (we will define several mortal precondition
operators in Section~\ref{sec:mp-operator-design}).  We
compute a mortal precondition for a program by traversing its $\omega$-path
expression DAG from the bottom up, using $\TF$ to interpret regular
expression operators and $\MP$ to interpret $\omega$-regular
expression operators.

We illustrate a selection of the interpretation steps.  We use
$\tfsem{-}$ and $\mpsem{-}$ to denote the interpretation of a regular and $\omega$-regular expression, respectively.
For the leaves
of the path expression DAG, we may simply encode the meaning of the
corresponding program command into logic; e.g., the transition
formulas for the edges $\tuple{c,d}$ and $\tuple{d,e}$ (corresponding
to the commands \cinline{[n >= 0]} and \cinline{m := m + 1},
resp.) are:
\begin{align*}
  \tfsem{\tuple{c,d}} &= n \geq 0 \land m' = m \land n' = n \land \cinline{step}' = \cinline{step}\\
  \tfsem{\tuple{d,e}} &= m' = m + 1\land n' = n \land \cinline{step}' = \cinline{step}
\end{align*}

Proceeding up the DAG, we compute a transition formula for the regular
expression $\tuple{c,d}\tuple{d,e}$ by taking the relational
composition of $\tfsem{\tuple{c,d}}$ and $\tfsem{\tuple{d,e}}$
\begin{align*}
  \tfsem{\tuple{c,d}\tuple{d,e}} &= \tfsem{\tuple{c,d}} \circ \tfsem{\tuple{d,e}}\\
  &\equiv
  \begin{array}{l@{\hspace*{1pt}}l}
    & n \geq 0\\
    \land & m' = m +1 \land n' = n \land \cinline{step}' = \cinline{step}
  \end{array}
\end{align*}
Similarly, we sequence $\tfsem{\tuple{c,d}\tuple{d,e}}$ with
$\tfsem{\tuple{b,c}}$ on the left and $\tfsem{\tuple{e,b}}$ on the
right to get a summary for the body of the inner loop $\textit{inner}
\defeq \tuple{b,c}\tuple{c,d}\tuple{d,e}\tuple{e,b}$:
\[ \tfsem{\textit{inner}} \equiv
\begin{array}{l@{\hspace*{1pt}}l}
  & m < \cinline{step} \land n \geq 0\\
  \land& m' = m + 1 \land n' = n - 1 \land \cinline{step}' = \cinline{step}
\end{array}
\]
The $\textit{inner}$ node has two parents, corresponding to
$\textit{inner}^*$ and $\textit{inner}^\omega$.  For the first, we
over-approximate the transitive closure of the formula $\tfsem{\textit{inner}}$:

\[\tfsem{\textit{inner}^*} \equiv \exists k.
\tightparen{\begin{array}{l@{\null}l}
  & \tightparen{\begin{array}{l@{\null}l}
    & k = 0\\
    \lor & \tightparen{\begin{array}{l@{\null}l}
      & k \geq 1 \land m < \cinline{step} \land n \geq 0\\
      \land & m' \leq \cinline{step} \land n' \geq -1
    \end{array}}
  \end{array}}\\
  \land & m' = m + k \land n' = n - k \land \cinline{step}' = \cinline{step}
\end{array}}\]
(In the above formula, the existentially quantified variable $k$ represents the
number of times the loop is taken.  The first conjunct encodes that if the loop
is taken at least once, then its guard must hold in the initial state, and the
post-image of its guard must hold in the final state.  The second conjunct
encodes that $m$ increases by 1 at each iteration, $n$ decreases by 1, and
\cinline{step} is constant. See Section~\ref{sec:transition-formulas} for
details on how we compute the transitive closure of any transition formula.)

For the second parent, $\textit{inner}^\omega$, we compute a mortal precondition for the formula
$\tfsem{\textit{inner}}$.  Observing that $(\cinline{step} - m)$ is a
ranking function for this loop (i.e., the difference between $\cinline{step}$ and $\cinline{m}$ is
non-negative and decreasing), we may simply take
$\mpsem{\textit{inner}^\omega} = \true$:
the inner loop terminates starting from any state.

Now consider the $\omega$-node corresponding to
the outer loop, $\textit{outer} = \tuple{a,b}\textit{inner}^*\tuple{b,a}$.  This loop illustrates a trade-off of
compositionality.  On one hand, compositionality makes proving
termination easier: by the time that we reach the $\omega$-node, we
have already built a transition formula that summarizes the body of
the outer loop.  Despite the fact that the body contains an inner
loop, we can use a theorem prover to answer questions about its
behavior (conservatively, since the summary is an over-approximation).  On
the other hand, compositionality makes termination proving more
difficult: a compositional analysis cannot prove that $n$
\textit{decreases} at each iteration, since it does not have access to
the surrounding context of the loop that initializes \cinline{step} to
\cinline{8}.  In Section~\ref{sec:mp-phase-analysis} we provide a
method that (for this particular loop) effectively performs a case split on
whether $n$ increases, decreases, or remains constant, and generates a mortal
precondition that is sufficient for all three cases: $\mpsem{\textit{outer}^\omega} = \cinline{step} > 0$.  Thus,
we have a \textit{conditional} termination argument: the
outer loop terminates as long as it begins in a state where \cinline{step}
is positive.

Continuing up the DAG, we combine the mortal preconditions of the inner and outer loops to get
\[ \mpsem{\textit{outer}^\omega + \textit{outer}^*\textit{inner}^\omega} \equiv \cinline{step} > 0\ . \]
Finally, we compute a mortal precondition for the root of the DAG (and
thus the whole program) by taking the weakest precondition of
$\cinline{step} > 0$ under the transition formula $\tfsem{\tuple{r,a}}
= \cinline{step}' = 8 \land m' = m \land n' = n$, yielding the
  formula $\true$.  Thus, by propagating the conditional termination argument for the outer loop backwards through its context, the analysis discharges the assumption of the conditional termination argument,
  and  verifies that the program always terminates.

\section{Background} \label{sec:background}

\subsection{Flow Graphs and Path Expressions} \label{sec:flow-graphs}

A \textbf{control flow graph} $G = \tuple{V,E,r}$ consists of a set of
vertices $V$, a set of directed edges $E \subseteq V \times V$, and a
root vertex $r \in V$ with no incoming edges.  A
\textbf{path} in $G$ is a finite sequence $e_1e_2\dots e_n \in E^*$
such that for each $i$, the destination of $e_i$ matches the source of
$e_{i+1}$; an \textbf{$\omega$-path} is an infinite sequence
$e_1e_2\dots \in E^\omega$ such that any finite prefix is a path.  For
any vertices $u,v \in V$, we use $\paths{G}{u}{v}$ to denote the (regular) set
of paths in $G$ from $u$ to $v$, and we use $\opaths{G}{u}$ to denote the (regular)
set of $\omega$-paths in $G$ starting from $u$.

We say that a vertex $u$ \textbf{dominates} a vertex $v$ if every path
from $r$ to $v$ includes $u$.  Every vertex dominates itself; we say
$u$ \textbf{strictly dominates} $v$ if $u$ dominates $v$ and $u \neq
v$.  We say that $u$ is the \textbf{immediate dominator} of $v$ if it
is the unique vertex that strictly dominates $v$ and is dominated by
every vertex that strictly dominates $v$.  The immediate dominance
relation forms a tree structure with $r$ as the root; we use
$\textit{children}(v)$ to denote the set of vertices whose immediate
dominator is $v$.  We say that $G$ is \textbf{reducible} if every
cycle contains an edge $\tuple{u,v}$ such that $v$ dominates $u$.

Taking the alphabet $\Sigma$ to be the set of edges in a given control flow
graph $G$, a regular set of (finite) paths in $G$ can be represented by a
regular expression, and a regular set of $\omega$-paths in $G$ can be recognized
by an $\omega$-regular expression; we call such regular expressions
($\omega$-)\textbf{path expressions}. The syntax of regular ($\RegExp(\Sigma)$)
and $\omega$-regular ($\oRegExp(\Sigma)$) expressions over an alphabet $\Sigma$
is given by (see e.g.  \cite{Book:POMC}, Ch. 4):
\begin{align*}
  a \in \Sigma\\
  e \in \RegExp(\Sigma) & ::= a \mid 0 \mid 1 \mid e_1 + e_2 \mid e_1e_2 \mid e^*\\
  f \in \oRegExp(\Sigma) & ::= e^\omega \mid ef \mid f_1 + f_2
\end{align*}
where 0 recognizes the empty language, 1 recognizes the empty word, $+$
corresponds to union, juxtaposition (or $\cdot$) to concatenation, $*$
to unbounded repetition, and $\omega$ to infinite repetition.  

\subsection{Logic and Geometry} \label{sec:logic}

The syntax of linear integer arithmetic (LIA) is given
as follows:
\begin{align*}
  x \in \textsf{Variable}\\
  n \in \mathbb{Z}\\
  t \in \textsf{Term} & ::= x  \mid n \mid n \cdot t \mid t_1 + t_2\\
  F \in \textsf{Formula} & ::= t_1 \leq t_2 \mid  t_1 = t_2  \mid F_1 \land F_2 \mid F_1 \lor F_2 \mid \lnot F \\ &\quad \, \mid \exists x. F \mid \forall x. F
\end{align*}

Let $X \subseteq \textsf{Variable}$ be a set of variables.  A \textbf{valuation}
over $X$ is a map $v : X \rightarrow \mathbb{Z}$.  If $F$ is a formula
whose free variables range over $X$ and $v$ is a valuation over $X$,
then we say that $v$ satisfies $F$ (written $v \models F$) if the
formula $F$ is true when interpreted over the standard model of the
integers, using $v$ to interpret the free variables.  We write $F
\models G$ if every valuation that satisfies $F$ also satisfies $G$.

For a formula $F$, we use $F[x \mapsto t]$ to denote the formula
obtained by substituting each free occurrence of the variable $x$ with the
term $t$.  We use the same notation to represent parallel substitution
of multiple variables by multiple terms; e.g., if $X$ is a set of
variables and $X' = \{ x' : x \in X \}$ is a set of ``primed''
versions of those variables, then $F[X \mapsto X']$ denotes the result
of replacing each variable in $x$ with its corresponding $x'$.
Substitution binds more tightly than logical connectives, so e.g., in
the formula $F \land G[x\mapsto y]$, $x$ is replaced with $y$ within
$G$, but not within $F$.

Let $F$ be an LIA formula with free variables $\vec{x} = x_1,\dots,x_n$.  The \textbf{convex hull} of $F$,
denoted $\conv(F)$, is the strongest (unique up to equivalence)
formula of the form $A\vec{x} \geq \vec{b}$ that is entailed by $F$, where $A$ is an integer matrix and $\vec{b}$ is an integer vector.
\citet{FMCAD:FK2015} give an algorithm for computing $\conv(F)$.

%% $F$ can be
%% considered as a representation of a set of points $\modelset(F)
%% \subseteq \mathbb{Z}^n$:
%% \begin{align*}
%%     \modelset(F) \defeq \{ &\vec{u} \in \mathbb{Z}^n : v_{\vec{u}} \models F \text{ where } \\ &v_{\vec{u}} \text{ is the valuation with } v_{\vec{u}}(x_i) = u_i \}
%% \end{align*}
%% It will
%% often be convenient to consider $\modelset(F)$ to be a subset of
%% $\mathbb{Q}^n$ (which happens to only contain integral points).

%% A set $C \subseteq \mathbb{Q}^n$ is called \textbf{convex} if for
%% every pair $\vec{u},\vec{v} \in C$, $C$ contains all points on the
%% line segment connecting $\vec{u}$ and $\vec{v}$.    A set $P
%% \subseteq \mathbb{Q}^n$ is called a \textbf{convex polyhedron} if it
%% is of the form $P = \{ \vec{v} \in \mathbb{Q}^n : A\vec{v} \geq
%% \vec{b} \}$ (with $A$ and $\vec{b}$ as above).  Let $F$ be a formula
%% with $n$ free variables.   The \textbf{convex hull} of
%% $F$, denoted $\conv(F)$, is the smallest convex set $\conv(F)
%% \subseteq \mathbb{Q}^n$ that contains $\modelset(F)$.  By writing $F$
%% in disjunctive normal form we can see that $\modelset(F)$ is a finite
%% union of integer solutions to systems of inequalities; therefore
%% $\conv(F)$ is a convex polyhedron.  \citet{FMCAD:FK2015} give an
%% algorithm for computing $\conv(F)$.

\subsection{Transition Formulas}
\label{sec:transition-formulas}
Fix a finite set $\Var$ of variables, and let $\Var' = \{ x'
: x \in \Var \}$ denote a set of ``primed copies'', presumed to be
disjoint from $\Var$.  A \textbf{state formula} is an LIA
formula whose free variables range over $\Var$.  A \textbf{transition
  formula} is an LIA formula whose free variables range
over $\Var \cup \Var'$.  We use $\SF{}$ and $\TF{}$ to denote sets
of state and transition formulas, respectively.  Define a
\textit{state} to be a valuation over $\Var$ (the set of which we denote \textsf{State}) and a \textit{transition}
to be a valuation over $\Var \cup \Var'$.  
Any pair of states $s,s'$
defines a transition $[s,s']$ which interprets each $x \in \Var$ as
$s(x)$ and each $x' \in \Var'$ as $s'(x)$.  A transition formula $F$
defines a relation $\rightarrow_F$ on states, with $s
\rightarrow_F s' \iff [s,s'] \models F$.

Define the \textit{relational composition} of two transition formulas
to be the formula
\[ F_1 \circ F_2 \defeq \exists \Var''. F_1[\Var' \mapsto \Var''] \land F_2[\Var \mapsto \Var'']\ . \]
For any $k \in \mathbb{N}$, we use $F^k$ to denote the $k$-fold
relational composition of $F$ with itself.  For a transition formula $F$ and a state formula $S$, define the weakest precondition of $S$ under $F$ to be the formula
\[ \wp(F,S) \defeq \forall \textsf{Var}'. F \Rightarrow S[\textsf{Var} \mapsto \textsf{Var}']\ . \]

We suppose the existence of an operation $(-)^\star$ that over-approximates the
reflexive transitive closure of a transition formula (i.e., for any transition
formula $F$, we have $\rightarrow_F^* \subseteq
\rightarrow_{F^\star}$).  Several such operators exist
\cite{FMCAD:FK2015,PACMPL:KCBR2018,POPL:KBCR2019,POPL:CBKR2019,CAV:SK2019};
here we will describe one such method, based on techniques from
\cite{ENTCS:ACI2010,FMCAD:FK2015}.

Let $\vec{x}'$ and $\vec{x}$ be vectors containing the variables $\Var'$
and $\Var$, respectively; let $n = |\Var|$ be the dimension of these
vectors.  In general, the transitive closure of a transition formula
cannot be expressed in first-order logic.  Two special cases where the transitive closure can be expressed are:
\begin{enumerate}
\item If $F$ takes the form $\textit{pre} \land \textit{post}$, where
  the free variables of $\textit{pre}$ range over $\Var$ and the free
  variables of $\textit{post}$ range over $\Var'$, then $F$ is already
  transitively closed, so we need only to take its reflexive closure:
  $(\textit{pre} \land \textit{post}) \lor \left(\bigwedge_{x \in
    \Var} x' = x\right)$
\item If $F$ takes the form $A\vec{x}' \geq A\vec{x} + \vec{b}$, then for
  any $k \in \mathbb{N}$, we have that $F^k$ is equivalent to
  $A\vec{x}' \geq A\vec{x} + k\vec{b}$, and so the formula $\exists k. k
  \geq 0 \land A\vec{x}' \geq A\vec{x} + k\vec{b}$ represents the reflexive transitive
  closure of $F$.
\end{enumerate}

Let $F$ be a transition formula.  We cannot expect $F$ to take one of
the above forms, but we \textit{can} always over-approximate $F$ by a
formula that does:
\begin{enumerate}
\item Let $\textit{Pre}(F) \defeq \exists \Var'.F$ and let
  $\textit{Post}(F) \defeq \exists \Var.F$.  We have that $F \models
  \textit{Pre}(F) \land \textit{Post}(F)$, and $\textit{Pre}(F) \land
  \textit{Post}(F)$ takes form (1) above.
\item For each variable $x$, let $\delta_x$ denote a fresh variable
  which we use to represent the difference between $x'$ and $x$; we
  use $\vec{\delta}$ to denote a vector containing the $\delta_x$
  variables.  The convex hull
  \[\conv\left(\exists
  \Var,\Var'. F \land \bigwedge_{x \in \Var} \delta_x = x' - x\right)
  \]
  takes the form $A\vec{\delta} \geq \vec{b}$.  Then we have $F
  \models A\vec{x}' \geq A\vec{x} + \vec{b}$, and $A\vec{x}' \geq
  A\vec{x} + \vec{b}$ takes form (2) above.
\end{enumerate}
Combining (1) and (2), we define an operation $\exp$ by
\begin{align*}
    \exp(F,k) \defeq &\left(\left( \bigwedge_{x \in \Var} x' = x\right) \lor (\textit{Pre}(F) \land \textit{Post}(F))\right) \\ 
    &\land A\vec{x}' \geq A\vec{x} + k\vec{b}
\end{align*}
and observe that for any $k \in \mathbb{N}$, we have that $F^k \models
\exp(F,k)$.  Finally, we over-approximate transitive closure by
existentially quantifying over the number of loop iterations:
\[ F^\star \defeq \exists k. k \geq 0 \land \exp(F,k)\ . \]
\begin{lemma} \label{lem:exp-monotone}
  The $(-)^\star$ and $\exp$ operators are monotone in the
  sense that if $F \models G$, then $F^\star \models G^\star$ and
  $\exp(F,k) \models \exp(G,k)$ (where $k$ is a
  variable symbol).
\end{lemma}

\subsection{Transition Systems}
A \textbf{transition system} $T$ is a pair $T = \tuple{S_T, R_T}$
where $S_T$ is a set of states and $R_T \subseteq S_T \times S_T$ is a
transition relation.  We write $\tr{T}{s}{s'}$ to denote that the pair
$\tuple{s,s'}$ belongs to $R_T$.  We say that a state $s \in S_T$ is
\textbf{mortal} if there exists no infinite sequence $s \rightarrow_T s_1 \rightarrow_T
s_2 \rightarrow_T s_3 \dots$.  A \textbf{mortal precondition} for $T$ is a state 
formula such that any state that satisfies the formula is mortal.
%; we say that $T$ is \textbf{well-founded} if every state is mortal.

Each transition formula $F$ defines a transition system, where the
state space is $\textsf{State}$, and where the transition relation is
$\rightarrow_F$.  Define a \textbf{mortal precondition operator} to be
a function $\mp : \TF \rightarrow \SF$, which given a transition formula
$F$, computes a state formula $\mp(F)$ that is a mortal precondition for 
$F$.  We say that $\mp$ is \textbf{monotone} if
for any transition formulas $F_1,F_2$ with $F_1 \models F_2$, we have
$\mp(F_2) \models \mp(F_1)$ (Note that this definition is
  \textit{antitone} with respect to the entailment ordering,
  but since \textit{weaker} mortal preconditions are more desirable it
  is natural to order mortal preconditions by reverse entailment.)

\begin{example} \label{ex:LLRF}
  \citet{PLDI:GMR2015} give a complete method for synthesizing linear
  lexicographic ranking functions (LLRFs) for transition formulas.  We
  may define a monotone mortal precondition operator
  $\mp_{\textit{LLRF}}$ as follows:
  \[ \mp_{\textit{LLRF}}(F) \defeq \begin{cases}
      \true & \text{if there is an LLRF for } F\\
      \neg \textit{Pre}(F) & \text{otherwise}
  \end{cases}\]
  The fact that $\mp_{\textit{LLRF}}$ is monotone follows from
  the fact that if $F_1 \models F_2$ then 
  $\textit{Pre}(F_1) \models \textit{Pre}(F_2)$ and 
  any LLRF for $F_2$ is also an LLRF for $F_1$, 
  and the completeness of \citet{PLDI:GMR2015}'s
  LLRF synthesis procedure.
\end{example}

Within this paper, a program is represented as a \textbf{labeled
  control flow graph} $P = \tuple{G,\elbl}$, where $G = \tuple{V,E,r}$
is a control flow graph, and $\elbl : E \rightarrow \TF{}$ is a
function that labels each edge with a transition formula.  $P$ defines
a transition system $\textit{TS}(P)$ where the state space is $V \times
\textsf{State}$, and where $\tr{P}{\tuple{v_1,s_1}}{\tuple{v_2,s_2}}$
iff $\tuple{v_1,v_2} \in E$ and $[s_1,s_2] \models \elbl(v_1,v_2)$.

\section{An Efficient \texorpdfstring{$\omega$}{w}-Path Expression Algorithm} \label{sec:pathexp}

This section describes an algorithm for computing an
$\omega$-regular expression that recognizes all infinite paths in a
graph that start at a designated vertex.  The algorithm
is based on Tarjan's path expression algorithm, which computes path
expressions that recognize \textit{finite} paths that start at a
designated vertex \cite{JACM:Tarjan1981}.  Our algorithm operates in
$O(|E|\alpha(|E|) + t)$ time, where $\alpha$ is the inverse Ackermann
function and $t$ is technical parameter that is $O(|V|)$ for reducible
flow graphs and is at most $O(|V|^3)$, matching the complexity of
Tarjan's algorithm.

It is technically convenient to formulate our algorithms on
\textit{path graphs} rather than control flow graphs.  A \textbf{path
  graph} for a flow graph $G = \tuple{V,E,r}$ is a graph $H =
\tuple{U,W}$ where $U \subseteq V$ and $W \subseteq U \times
\RegExp(E) \times U$ is a set of directed edges labeled by regular expressions over $E$, and such that for every $\tuple{u, e, v} \in W$, $e$
recognizes a subset of $\paths{G}{u}{v}$.  We say that $H$
\textit{represents} a path $p$ from $u$ to $v$ (in $G$, with $u,v \in
U$) if there is a path \[(w_1,e_1,w_2)(w_2,e_2,w_3) \dots
(w_{n},e_n,w_{n+1})\] in $H$ with $w_1 = u$ and $w_{n+1} = v$ and $p$
is recognized by the regular expression $e_1e_2 \dots e_n$.
Similarly, $H$ represents an $\omega$-path $p$ if there is a
decomposition $p = p_1p_2p_3 \dots$ and an $\omega$-path
$(w_1,e_1,w_2)(w_2,e_2,w_3) \dots$ in $H$ with $p_i$ recognized by
$e_i$ for all $i$.  We use $\pathrep{H}{u}{v}$ to denote the set of
paths from $u$ to $v$ that are represented by $H$, and
$\opathrep{H}{v}$ to denote the set of $\omega$-paths starting at $v$
that are represented by $H$.  We say that $H$ is \textbf{complete}
for a set of edges $E' \subseteq E$ if
\begin{enumerate}
\item For each $u,v \in U$, $\pathrep{H}{u}{v}$ is the set of paths from $u$
  to $v$ in $G$ consisting only of edges from $E'$.
\item For each $v \in U$, $\opathrep{H}{v}$ is the set of
  $\omega$-paths from $v$ in $G$ consisting only of edges from $E'$,
  and which visit some vertex of $U$ infinitely often.
\end{enumerate}

\subsection{A Na\"{i}ve Algorithm}
Algorithm~\ref{alg:solve-dense} is a na\"{i}ve algorithm 
for computing path expressions, which is used as a sub-procedure in
the main algorithm.  It is a variation of the classic state
elimination algorithm for converting finite automata to regular
expressions.  The input to Algorithm~\ref{alg:solve-dense} is a path
graph $H = \tuple{U,W}$ (for some flow graph $G$) and a root vertex
$r$ (not necessarily the root of $G$); its output is a pair consisting
of an $\omega$-path expression that recognizes $\opathrep{H}{r}$ and a
function that maps each vertex $v \in U$ to a path expression that
recognizes $\pathrep{H}{r}{v}$.  The idea is to
successively eliminate the outgoing edges of every vertex in the graph
except the root, while preserving the set of paths (and
$\omega$-paths) emanating from vertices whose outgoing edges have not
yet been removed.  The algorithm operates in $O(|U|^3)$ time.

\begin{algorithm}
  \caption{Na\"{i}ve path expression algorithm \label{alg:solve-dense}}
  \Fn{$\textsf{solve-dense}(H,r)$}{
  \Input{Path graph $H = \tuple{U,W}$, vertex $r \in U$ with no incoming edges}
  \Output{Pair $\tuple{\textit{pe}^\omega,\textit{pe}}$ where
 $\textit{pe}^\omega \in \oRegExp(E)$ recognizes $\opathrep{H}{r}$ and
 $\textit{pe} : U \rightarrow \RegExp(E)$ maps each $v \in U$ to a path expression that recognizes $\pathrep{H}{r}{v}$.}
  \tcc{$\textit{pe}(u,v)$ recognizes paths from $u$ to $v$ represented by $H$}
  $\textit{pe} \gets \lambda \tuple{u,v}. 0$\;
  \ForEach{$\tuple{u,e,v} \in W$}{
    $\textit{pe}(u,v) \gets \textit{pe}(u,v) + e$\;
  }
  \tcc{Suppose $V$ is ordered as $V = \{r=v_0,v_1,\dots,v_n\}$}
  \For{$i$ = $n$ downto 1}{
%    \Inv{$\tuple{U,\set{ \tuple{v_t,p(v_t,u),u} : t < i } \cup \{ \tuple{v_t,p(v_t,v_t),v_t} : 0 \leq t \leq n \}}$ represents the same $(\omega-)$paths as $H$}
    \For{$j$ = $i-1$ downto $0$}{
      $e_{ji} \gets \textit{pe}(v_j,v_i) \cdot \textit{pe}(v_i,v_i)^*$\;
      \For{$k$ = $n$ downto $1$, $k \neq i$}{
        $\textit{pe}(v_j,v_k) \gets \textit{pe}(v_j,v_k) + e_{ji} \cdot \textit{pe}(v_i,v_k)$
      }
    }
  }
  \Return{$\tuple{\begin{array}{l}\sum_{i=1}^n \textit{pe}(r,v_i) \cdot \textit{pe}(v_i,v_i)^\omega,\\ \lambda v. \textit{pe}(r,v) \cdot \textit{pe}(v,v)^*\end{array}}$}
  }
\end{algorithm}

\subsection{\texorpdfstring{$\omega$}{w}-Path Expressions in Nearly Linear Time}

\begin{figure*}
  \begin{subfigure}[t]{5cm}
\begin{tikzpicture}[thick,>=stealth,
      v/.style={circle,draw,minimum height=18pt},
      node distance=1.5cm]

  \node [v] (r) {$1$};
  \node [v,below of=r,xshift=-2cm] (a) {$2$};
  \node [v,below of=r] (b) {$3$};
  \node [v,below of=r,xshift=2cm] (c) {$4$};
  \node [v,below of=b] (d) {$5$};

  \draw (r) edge[->] node[above left]{$a$} (a);
  \draw (r) edge[->] node[above right]{$b$} (c);
  \draw (a) edge[->] node[above]{$d$} (b);
  \draw (b) edge[->] node[left]{$f$} (d);
  \draw (c) edge[->] node[above]{$e$} (b);
  \draw (d) edge[->] node[below right]{$g$} (c);
  \draw (a) edge[->,looseness=4,out=270,in=200] node [below right]{$c$} (a);

\end{tikzpicture}
\caption{A control flow graph \label{fig:pe-cfg}}
  \end{subfigure}
  \begin{subfigure}[t]{3cm}
  \centering
\begin{tikzpicture}[thick,>=stealth,
      v/.style={},
      node distance=1.5cm]
        
  \node [v] (r) {$1$};
  \node [v,below of=r,xshift=-1cm] (a) {$2$};
  \node [v,below of=r,xshift=0cm] (b) {$3$};
  \node [v,below of=r,xshift=1cm] (c) {$4$};
  \node [v,below of=b] (d) {$5$};

  \draw (r) edge[->] (a);
  \draw (r) edge[->] (b);
  \draw (r) edge[->] (c);
  \draw (b) edge[->] (d);
\end{tikzpicture}
\caption{Dominator tree \label{fig:pe-dt}}
  \end{subfigure}
\begin{subfigure}[t]{3cm}
  \centering
\begin{tikzpicture}[thick,>=stealth,
      v/.style={},
      node distance=1.5cm]
  \node [v] (r) {$1$};
  \node [v,below of=r,xshift=-1cm] (a) {$2$};
  \node [v,below of=r,xshift=0cm] (b) {$3$};
  \node [v,below of=r,xshift=1cm] (c) {$4$};
  \node [v,below of=b,xshift=-1cm] (d) {$5$};
  \node [right of=d,xshift=2pt,fill=gray!5,inner sep=0] (l0) {\footnotesize$\begin{array}{r@{\null}l}
      E_0 &\defeq ac^*d+be\\
      E_1 &\defeq E_0\cdot(fge)^*\\ 
      E_2 &\defeq b+E_1\cdot fg
  \end{array}$};
  \draw (r) edge[->,dashed] node[left]{$ac^*$} (a);
  \draw (r) edge[->,dotted] node[fill=white] {$E_1$} (b);
  \draw (r) edge[->,dotted, right] node{$E_2$} (c);
  \draw (b) edge[->] node[left]{$f$} (d);

\end{tikzpicture}
\caption{Weighted forest \label{fig:pe-forest}}
\end{subfigure}
  \begin{subfigure}[t]{2.5cm}
  \centering
\begin{tikzpicture}[thick,>=stealth,
      v/.style={},
      node distance=1.5cm]
        
  \node [v] (a) {$2$};
  \node [v,below of=a] (b) {$3$};
  \node [v,below of=b] (c) {$4$};

  \draw (a) edge[->] (b);
  \draw (b) edge[->,bend left] (c);
  \draw (c) edge[->,bend left] (b);
  \draw (a) edge[->,looseness=4,out=270,in=200] (a);
\end{tikzpicture}
\caption{$\textit{SiblingGraph}(1)$ \label{fig:pe-sibling-graph}}
  \end{subfigure}
\begin{minipage}[b]{3.75cm}
\begin{subfigure}[t]{3.75cm}
  \centering
\begin{tikzpicture}[thick,>=stealth,
      v/.style={},
      node distance=1.5cm]
        
  \node [v] (r) {$1$};
  \node [v,right of=r] (a) {$2$};
  \draw (r) edge[->] node[above]{$a$} (a);
  \draw (a) edge[->,looseness=4,out=290,in=350] node [below]{$c$} (a);
\end{tikzpicture}
\caption{$\textit{ComponentGraph}(\{2\})$ \label{fig:pe-component-graph2}}
\end{subfigure}
\begin{subfigure}[t]{3.75cm}
  \centering
\begin{tikzpicture}[thick,>=stealth,
      v/.style={},
      node distance=1.5cm]
        
  \node [v] (r) {$1$};
  \node [v,below left of=r] (b) {$3$};
  \node [v,below right of=r] (c) {$4$};

  \draw (r) edge[->] node [above left]{\small$ac^*d$} (b);
  \draw (r) edge[->] node [above right]{$b$}(c);
  \draw (b) edge[<-,bend left=5] node [above]{$e$} (c);
  \draw (c) edge[<-,bend left] node [below]{$fg$} (b);
\end{tikzpicture}
\caption{$\textit{ComponentGraph}(\{3,4\})$ \label{fig:pe-component-graph34}}
\end{subfigure}
\end{minipage}
\caption{Operation of \textsf{solve-sparse} on an example control flow graph.}
\label{fig:pe}
\end{figure*}
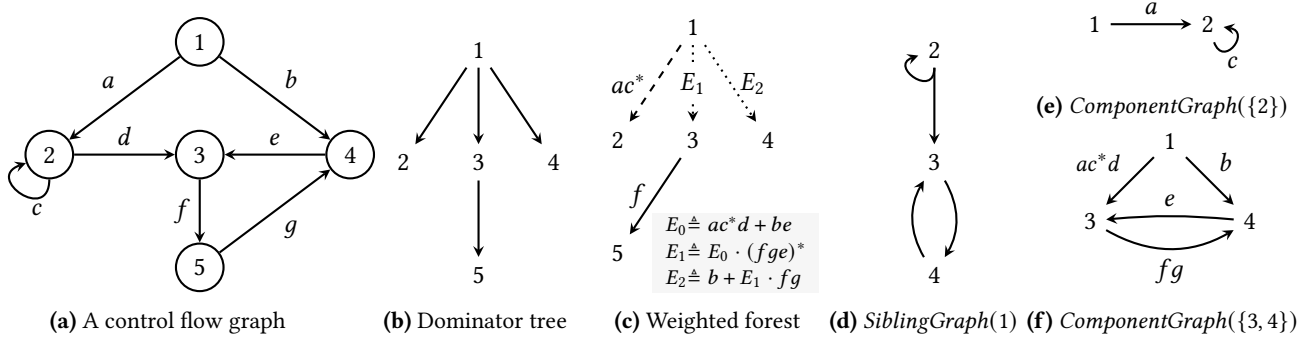

Algorithm~\ref{alg:solve-sparse} is an efficient $\omega$-path expression algorithm
that exploits sparsity of control
flow graphs.  Following Tarjan \cite{JACM:Tarjan1981}, the algorithm uses the dominator tree of
the graph to break it into single-entry components, and uses
a compressed weighted forest data structure to combine paths from different
components.

A \textbf{compressed weighted forest} is a data structure that
represents a forest of vertices with edges weighted by regular
expressions.  The data structure supports the following operations:
\begin{itemize}
\item $\flink(u,e,v)$: set $v$ to be the parent of $u$ by adding an edge from $v$ to $u$ labeled $e$ ($u$ must be a root)
\item $\ffind(v)$: return the (unique) vertex $u$ such that $u \rightarrow^* v$ and $u$ is a root
\item $\feval(v)$: return the regular expression $e_1\dotsi e_n$, where $u_1 \xrightarrow{e_1}  u_2 \xrightarrow{e_2} \dotsi \xrightarrow{e_n} v$ is the path from a root to $v$.
\end{itemize}
This data structure can be implemented so that each operation takes
$O(\alpha(n))$ amortized time, where $n$ is the number of vertices in
the forest \cite{JACM:Tarjan1979}.

%% Intuitively, \textsf{solve-sparse} traverses the dominator tree for
%% $G$ bottom-up.  Each vertex $v$ in the dominator tree identifies a
%% single-entry region, consisting of those vertices that are dominated
%% by $v$.

The subroutine $\textsf{solve-sparse}(v)$ returns an $\omega$-path
expression that recognizes the set of $\omega$-paths $\opaths{G}{v}
\cap E|_v^\omega$, where \[ E|_v \defeq \{ \tuple{u_1,u_2} \in E : u_2
\text{ is strictly dominated by } v\}\ .\] Moreover, it maintains the
invariant that after completing a call to $\textsf{solve-sparse}(v)$,
we have that for every vertex $u$ that is dominated by $v$, $\ffind(u)
= v$ and $\feval(u)$ is a path expression that recognizes
$\paths{G}{v}{u} \cap E|_v^*$.

The $\textsf{solve-sparse}(v)$ subroutine is structured as a recursive traversal of the dominator tree (see Example~\ref{ex:pe}).  First, it
 calls $\textsf{solve-sparse}(c)$ for each child $c$ of
$v$ in the dominator tree.  Next, it computes a directed graph $G_v =
\textit{SiblingGraph}(v)$ whose vertices are $\textit{children}(v)$,
and such that there is an edge $\tuple{c_1,c_2}$ iff there is a path
from $c_1$ to $c_2$ in $G$ such that each edge (except the last) belongs to $E|_{c_1}$.  The
edges of $\textit{SiblingGraph}(v)$ can be computed efficiently as
\[ \set{\tuple{\ffind(u), c} : c,\ffind(u) \in \textit{children}(v), \tuple{u,c} \in E}\ . \]
The correctness argument for the edge computation is as follows.  If
there is a path from $c_1$ to $c_2$ consisting of $E|_{c_1}$ edges, it
takes the form $\pi \tuple{u,c_2}$ for some $\tuple{u,c_2}$ in $E$,
with $\pi \subseteq E|_{c_1}^*$ a path from $c_1$ to $u$.  Since
$\pi$ ends at $u$ and consists only of $E|_{c_1}$-edges, $u$ is dominated by $c_1$.  Since $\textsf{solve-sparse}(v)$
 calls $\textsf{solve-sparse}(c_1)$ before constructing
$\textit{SiblingGraph}$, we must have $\ffind(u) = c_1$ by the
invariants of $\textsf{solve-sparse}$.

Next, $\textsf{solve-sparse}$ computes the strongly connected components of
$G_v$ and processes them in topological order.  The loop
(lines~\ref{ln:scc-loop-begin}-\ref{ln:scc-loop-end}) maintains the
invariant that when processing a component $C$, for every sibling node
$u$ that is topologically ordered before $C$, we have that $\ffind(u)
= v$ and that $\feval(u)$ recognizes $\paths{G}{v}{u} \cap E|_v^*$.
To process a component $C$, we form a path graph $G_C =
\textit{ComponentGraph}(C)$ whose vertices are $C \cup \set{v}$ and
that is complete for $E|_v$, computing a path expression
$\textit{C-pe}(u)$ for each $u \in C$ that recognizes $\paths{G}{v}{u}
\cap E|_v^*$ using \textsf{solve-dense}, and then linking $u$ to $v$
with the path expression $\textit{C-pe}(u)$ in the compressed weighted
forest.  The edges of $\textit{ComponentGraph}(C)$ are obtained by
collecting all weighted edges of the form
$\tuple{\ffind(w),\feval(w),u}$ such that $u \in C$ and $\tuple{w,u}
\in E$; the fact that $G_C$ is complete for $E|_v$ follows from the
loop invariant, using an argument analogous to the correctness
argument for the $\textit{SiblingGraph}$ construction above.  Finally,
we return an $\omega$-path expression which is the sum of
(line~\ref{ln:component-pe}) the $\omega$-path expressions for each
component and (line~\ref{ln:scc-loop-end}) an $\omega$-path expression
for each child $c$, pre-concatenated with a path expression recognizing
$\paths{G}{v}{c} \cap E|_v^*$.

\begin{example} \label{ex:pe}
Figure~\ref{fig:pe} illustrates the $\textsf{solve-sparse}$ procedure.  Figure~\ref{fig:pe-cfg}
depicts a control flow graph, whose dominator tree appears in
Figure~\ref{fig:pe-dt} (for legibility, we refer to edges by label rather than by their endpoints).  Consider the operation of
$\textsf{solve-sparse}(1)$.  The compressed weighted forest after
 calling $\textsf{solve-sparse}$ on $1$'s children
$2,3,4$ is depicted in Figure~\ref{fig:pe-forest} (the single solid link
from $3$ to $4$ labeled $f$; the other links are added later).  The sibling graph for
$1$ is given in Figure~\ref{fig:pe-sibling-graph} -- observe that it
has two strongly connected components: $\{2\}$ and $\{3,4\}$, with
$\{2\}$ ordered topologically before $\{3,4\}$.

The loop (lines~\ref{ln:scc-loop-begin}-\ref{ln:scc-loop-end}) processes $\{2\}$ first, producing the component graph
in Figure~\ref{fig:pe-component-graph2}.  Then $2$ is linked to $1$ in
the compressed weighted forest (dashed edge of
Figure~\ref{fig:pe-forest}) with the regular expression $ac^*$ (the
paths from $1$ to $2$ represented by
$\textit{ComponentGraph}(\{2\})$).

Next, the loop processes the $\{3,4\}$ component, producing the
component graph in Figure~\ref{fig:pe-component-graph34}; note that
the edge from $2$ to $3$ in $G$ produces the edge from $1$ to $3$
(since $\ffind(2) = 1$) and the edge from $5$ to $4$ produces the edge
from $3$ to $4$ (since $\ffind(5) = 3$).  Then $3$ and $4$ are both
linked to $1$ in the compressed weighted forest (dotted edges of Figure~\ref{fig:pe-forest}).

Finally, \textsf{solve-sparse} returns the sum
$\textsf{solve-dense}(G_{\{2\}}, 1)$ and
$\textsf{solve-dense}(G_{\{3,4\}}, 1)$, which is the $\omega$-path
expression
\[
(ac^*d + be)(fge)^\omega + ac^\omega \let\qedsymbol\romanqed\qedhere
\]
\end{example}

Algorithm~\ref{alg:solve-sparse} operates in $O(|E|\alpha(|E|) + t)$
time, where $t$ is the time taken by the calls to
$\textsf{solve-dense}$.  For reducible flow graphs, each sibling
graph is a singleton (see \cite{JACM:Tarjan1981}), so the complexity
simplifies to $O(|E|\alpha(|E|))$.

\begin{algorithm}
  \Alg{$\opathexp{G}{r}$}{
    \Input{Graph $G = \tuple{V,E}$ and root vertex $r$}
    \Output{An $\omega$-path expression recognizing $\opaths{G}{r}$}
    $\textit{children} \gets$ dominator tree for $G$\;
    Init compressed weighted forest with vertices $V$\;
    \Return{\textsf{solve-sparse}$(r)$}

  \Fn{$\textsf{solve-sparse}(v)$}{
    \Input{Vertex $v \in V$}
    \Output{An $\omega$-path expression recognizing $\opaths{G}{v} \cap E|_v^\omega$}
    \ForEach{child $c \in \textit{children(v)}$}{
      $\textit{child-pe}_\omega(c) \gets \textsf{solve-sparse}(c)$\;
    }
    $G_v \gets \textit{SiblingGraph}(v)$\;
    %\tuple{ \textit{children(v)}, \set{\tuple{\ffind(p), c} : c,\ffind(p) \in \textit{children}(v), \tuple{p,c} \in E}}$\;
    $\textit{pe}_\omega \gets 0$ \tcc{accumulating $\omega$-path expression} \label{ln:scc-loop-begin}
    \ForEach{s.c.c. $C$ of $G_v$ in topological order}{
      %% \Inv{for all vertices $u$ preceeding $C$ in topological order,
      %%   $\ffind(u) = v$ and $\feval(u)$ recognizes $\paths{G}{v}{u} \cap E|_v^*$;
      %%   $f$ recognizes the subset of $\opaths{G}{v} \cap E|_v^*$
      %%   that does not pass through a vertex in $C$ or succeeding $C$ in topological order}
      $G_C \gets \textit{ComponentGraph}(C)$\;
%      $G_C \gets \tuple{ C \cup \set{v}, \{ \tuple{\ffind(p),\feval(p) \cdot \tuple{p,v},v} : \tuple{p,v} \in E, v \in C \}}$\;
%      \tcc{$G_C$ is path graph on $C \cup \set{v}$ complete for $E|_v$}
      $\tuple{\textit{C-pe}_\omega,\textit{C-pe}} \gets \textsf{solve-dense}(G_C, v)$\;
      $\textit{pe}_\omega \gets \textit{pe}_\omega + \textit{C-pe}_\omega$\; \label{ln:component-pe}
      \ForEach{$u \in C$}{
        $\flink(u,\textit{C-pe}(u),v)$\;
        $\textit{pe}_\omega \gets \textit{pe}_\omega + \textit{C-pe}(u) \cdot \textit{child-pe}_\omega(u)$\;
      }
    } \label{ln:scc-loop-end}
    \Return{$\textit{pe}_\omega$}
  }
  }
  \caption{An $\omega$-path expression algorithm \label{alg:solve-sparse}}
\end{algorithm}

\section{Algebraic Termination Analysis} \label{sec:ata}

This section describes the process of interpreting an
\mbox{($\omega$-)}regular expression within a suitable algebraic structure.
As a particular case of interest, we show how to apply the algebraic
framework to termination analysis.

An \textbf{interpretation} over an alphabet $\Sigma$ consists
of a triple $\interp = \tuple{\mathbf{A},\mathbf{B},\elbl}$, where
$\mathbf{A}$ is a \textit{regular algebra}, $\mathbf{B}$ is a
\textit{$\omega$-regular algebra over $\mathbf{A}$}, and
$\elbl : \Sigma \rightarrow \mathbf{A}$ is a
\textit{semantic function}.  A \textbf{regular algebra} $\mathbf{A} =
\tuple{A,0^A,1^A,+^A,\cdot^A,\null^{*^A}}$ is an algebraic structure
equipped with two distinguished elements $0^A,1^A \in A$, two binary
operations $+^A$ and $\cdot^A$, and a unary operation $(-)^{*^A}$.  An
\textbf{$\omega$-algebra} over $\mathbf{A}$ is 4-tuple $\mathbf{B} =
\tuple{B,\cdot^B,+^B,^{\omega^B}}$ consisting of a universe $B$, an
operation $\cdot^B : A \times B \rightarrow B$, an operation $+^B : B
\times B \rightarrow B$, and an operation $(-)^{\omega^B} : A
\rightarrow B$.  A \textbf{semantic function} $\elbl :
\Sigma \rightarrow A$ maps the letters of $\Sigma$ into the regular
algebra $\mathbf{A}$.

Given an interpretation $\interp =
\tuple{\mathbf{A},\mathbf{B},\elbl}$ over an alphabet $\Sigma$, we can
evaluate any regular expression $e$ over $\Sigma$ to an element
$\abssem{e}$ of $\mathbf{A}$ and any $\omega$-regular expression $f$
over $\Sigma$ to an element $\absosem{f}$ of $\mathbf{B}$ by
interpreting each letter according to the semantic function and each
($\omega$-)regular  operator using its corresponding
operator in $\mathbf{A}$ or $\mathbf{B}$:\\
\begin{minipage}{\linewidth}
{\noindent\null\hfill
  $\abssem{a} \defeq \elbl(a) \quad\textit{for $a \in \Sigma$}$
  \hfill\null}\\
\noindent\begin{minipage}{0.4\linewidth}
\small
\begin{align*}
  \abssem{0} &\defeq 0^A\\
  \abssem{1} &\defeq 1^A\\
  \abssem{e_1e_2} &\defeq \abssem{e_1} \cdot^A \abssem{e_2}\\
  \abssem{e_1 + e_2} &\defeq \abssem{e_1} +^A \abssem{e_2}\\
  \abssem{e^*} &\defeq \abssem{e}^{*^A}
\end{align*}
\end{minipage}\hfill
\begin{minipage}{0.5\linewidth}
\small
\begin{align*}
  \absosem{e^\omega} &\defeq \abssem{e}^{\omega^B}\\
  \absosem{ef} &\defeq \abssem{e} \cdot^B \absosem{f}\\
  \absosem{f_1 + f_2} &\defeq \absosem{f_1} +^B \absosem{f_2}
\end{align*}
\end{minipage}

\end{minipage}

If an $\omega$-path expression $f$ is represented by a DAG with $n$
nodes, we can process the DAG bottom-up (as in
Section~\ref{sec:overview}) to compute $\absosem{f}$ in $O(n)$
operations of $\mathbf{A}$ and $\mathbf{B}$.  

\subsection{Termination Analysis}

This paper is primarily concerned with applying the above
algebraic framework to termination analysis.  The fundamental
operation of interest is this setting the $\omega$-iteration operator.

Fix a mortal precondition operator $\mp : \TF \rightarrow \SF$.  We
define a regular algebra of transition formulas, $\TF$, and an
$\omega$-regular algebra of mortal preconditions, $\MP$.  The universe
of $\TF$ is the set of transition formulas, and the universe of $\MP$
is the set of state formulas.  The operations are given below:

\begin{minipage}{0.45\linewidth}
\begin{align*}
  0^{\TF} &\defeq \false\\
  1^{\TF} &\defeq \bigwedge_{x \in  \Var} x' = x\\
  F_1 +^{\TF} F_2 &\defeq F_1 \lor F_2\\
  F_1 \cdot^{\TF} F_2 &\defeq F_1 \circ F_2\\
  F^{*^\TF} &\defeq F^\star
\end{align*}
\end{minipage}
\begin{minipage}{0.45\linewidth}
\begin{align*}
  F^{\omega^{\MP}} &\defeq \mp(F) \\   
  F \cdot^{\MP} S & \defeq \wp(F,S) \\
  S_1 +^{\MP} S_2 & \defeq S_1 \land S_2 
\end{align*}
\end{minipage}

Let $P = \tuple{G, \elbl}$ be a labeled control flow graph, which
defines a transition system $\textit{TS}(P)$.  Let $r$ be the root of $G$.  Using
the algorithm in Section~\ref{sec:pathexp}, we can compute an
$\omega$-regular expression that recognizes all $\omega$-paths in $G$
beginning at $r$.  By interpreting this regular expression (as above)
under the interpretation $\mathcal{T} \defeq \tuple{\TF,\MP,\elbl}$, we can
under-approximate the mortal initial states of $\textit{TS}(P)$.  The
correctness of this strategy is formalized below.

\begin{restatable}[Soundness]{proposition}{soundness} \label{lem:ata-soundness}
  Let $P = \tuple{G, \elbl}$ be a labeled CFG, let $r$
  be the root of $G$, and let $\opathexp{G}{r}$ be an $\omega$-path
  expression recognizing $\opaths{G}{r}$.  Then
  $\mpsem{\opathexp{G}{r}}$ is a mortal precondition for $\textit{TS}(P)$,
  in the sense that for any $s \models \mpsem{\opathexp{G}{r}}$,
  we have that $\tuple{r,s}$ is a mortal state of $\textit{TS}(P)$.  In
  particular, if $\mpsem{\opathexp{G}{r}}$ is valid, then the
  program $P$ has no infinite executions.
\end{restatable}

\begin{restatable}[Monotonicity]{proposition}{monotonicity} \label{lem:ata-monotonicity}
  Suppose that $\mp$ is a monotone mortal precondition operator, and $\MP$ and $\TF$ are defined as above.
  Let $f \in \oRegExp(E)$, and let $L_1,L_2 : E
  \rightarrow \TF$ be semantic functions such that for all $e \in E$,
  $L_1(e) \models L_2(e)$.  Define $\mathcal{T}_1 \defeq
  \tuple{\TF,\MP,L_1}$ and $\mathcal{T}_2 \defeq \tuple{\TF,\MP,L_2}$.
  Then $\mathcal{T}^\omega_2\!\sem{f} \models
  \mathcal{T}^\omega_1\!\sem{f}$.
\end{restatable}

\subsection{Inter-procedural Analysis} \label{sec:interprocedural}

Our algebraic framework extends to the inter-procedural case using the
method of \citet{FMSD:CPR2009}.  The essential idea is to merge the
control flow graphs of all procedures of a program into an
\textit{inter-procedural control flow graph} (ICFG) so that infinite
paths through the program---including paths that are infinite due to
the presence of recursion---correspond to infinite paths its ICFG.  We
may then compute $\omega$-path expressions for the ICFG and interpret
them, just as in the intra-procedural case.  That is, the same analysis
that is used to prove conditional termination for loops also can be
applied to recursive functions.  In the following we sketch the
inter-procedural extension; see \cite{TR:ZK2021} for details.

We represent a multi-procedure program as a tuple \[P =
\tuple{V,E,\textit{Proc},\Lambda,\textit{entry},\textit{exit}}\ ,\]
where $\tuple{V,E}$ is a finite directed graph, $\textit{Proc}$ is a
finite set of procedure names, $\Lambda : E \rightarrow (\TF \cup
\textit{Proc})$ labels each edge by either a transition formula or a
procedure call, and $\textit{entry},\textit{exit} : \textit{Proc}
\rightarrow V$ are functions associating each procedure name with an
entry and an exit vertex.  We presume that the set of variables $\Var$
is divided into a set of local variables $\LVar$ and a set of global
variables $\GVar$.  Note that procedures do not have parameters or
return values, but these can be modeled using global variables (see
Figure~\ref{fig:fib} for an example)

Fix a program $P$. Define its \textit{inter-procedural} control flow
graph $\textit{ICFG} \defeq (V,E_{\textit{ICFG}})$. as follows.  The
vertices $V$ are the same as the vertices of $P$.  The edges
$E_{\textit{ICFG}} \defeq E \cup \textit{Interproc}$ are the edges of $P$ plus an
additional set of \textit{inter-procedural} edges, which represent
transfer of control between procedures by connecting the source of a
call to the entry of the called procedure:
\[ \textit{Interproc} \defeq \set{ \tuple{u,\textit{entry}(p)} : \exists \tuple{u,v} \in E. \Lambda(u,v) = p }\ . \]
An example \textit{ICFG} appears in Figure~\ref{fig:fib}; dashed edges
correspond to inter-procedural edges.

Finally, we define a semantic function that can be used to interpret
the edges of $\textit{ICFG}$.  A \textit{summary assignment} is a
function $\summary : \textit{Proc} \rightarrow \TF$ that maps each
procedure to a transition formula that over-approximates its behavior.
For example, one possible summary assignment for Figure~\ref{fig:fib}
is $\summary(\cinline{fib}) = g \leq r'$, indicating that the output
of \cinline{fib} is no less than its input.  Summary assignments can
be computed using standard iterative techniques (some care needs to be
taken to ensure monotonicity; see \cite{TR:ZK2021} for details).  With
a summary assignment $\summary$ in hand, we can define a semantic
function $L_\summary: E_{\textit{ICFG}} \rightarrow \mathbf{TF}$ by
\[
L_\summary(u,v) \defeq \begin{cases}
  \Lambda(u,v) & \text{if } \tuple{u,v} \in E  \text{ and } \Lambda(u,v) \in \TF\\
  \summary(p) & \text{if } \tuple{u,v} \in E \text{ and } \Lambda(u,v) = p\\
  {\displaystyle\bigwedge_{x \in \GVar} x' = x} & \text{if } \tuple{u,v} \in \textit{Interproc}
\end{cases} \]

\begin{theorem}[Inter-Procedural Soundness]
  Let $P$ be a program.  For any procedure $p \in P$,
  $\mpsem{\opathexp{\textit{ICFG}}{\textit{entry}(p)}}$ is a mortal precondition for
  the procedure $p$, in the sense that for any state $s$ such that
  $s \models
  \mpsem{\opathexp{\textit{ICFG}}{\textit{entry}(p)}}$, we have that
  $\tuple{\textit{entry}(p),s}$ is a mortal state of $P$.
\end{theorem}

\begin{example}
  Consider the recursive Fibonacci function and its inter-procedural control flow graph pictured in Figure~\ref{fig:fib}.  We have
  \begin{align*}
    \opathexp{\textit{ICFG}}{r} &= \textit{body}^\omega, \text{ where}\\
    \textit{body} &= \tuple{r,a}\tuple{a,b}(\tuple{b,r} + \tuple{b,c}\tuple{c,d}\tuple{d,r})
  \end{align*}
  Observe that any infinite execution of  \cinline{fib}
  corresponds to a path in its \textit{ICFG}, and therefore $\opathexp{\textit{ICFG}}{r}$.
  We can compute a precondition under which Fibonacci terminates by evaluating
  $\opathexp{\textit{ICFG}}{r}$, using $\mp_{\textit{LLRF}}$ as the mortal precondition operator:
\begin{align*}
  \tfsem{\textit{body}} &\equiv g \geq 2 \land (g' = g - 1 \lor g' = g - 2)\\
  \mpsem{\opathexp{\textit{ICFG}}{r}} &= \true
\end{align*}
\end{example}

\begin{figure}
  \begin{minipage}{7cm}
\begin{lstlisting}[style=base,numbers=left,xleftmargin=2em]
fib(n):
  if (n <= 1):
    return 1
  else
    return fib(n - 1) + fib(n - 2)
\end{lstlisting}
  \end{minipage}
  \hfill
  \begin{minipage}{6cm}
  \begin{tikzpicture}[thick,>=stealth,
      v/.style={circle,draw,minimum height=18pt},
      node distance=1.5cm]
    \small
  \node [v] (r) {$r$};
  \node [v,below of=r] (a) {$a$};
  \node [v,right of=a,yshift=-2cm,xshift=1cm] (x) {$x$};
  \node [v,below of=a] (b) {$b$};
  \node [v,below of=b] (c) {$c$};
  \node [v,below of=c] (d) {$d$};
  \node [v,right of=d,xshift=1cm] (e) {$e$};
  \draw (r) edge[->] node[left]{$n' = g$} (a);
  \draw (a) edge[->] node[above,xshift=7pt]{$\begin{array}{l@{\null}l}&n \leq 1\\ \land & r' = 1\end{array}$} (x);
  \draw (a) edge[->] node[left]{$\begin{array}{l@{\null}l} & n \geq 2\\ \land & n' = n\\ \land & g' = n - 1\end{array}$} (b);
  \draw (b) edge[->] node[left]{call \cinline{fib}} (c);
  \draw (c) edge[->] node[left]{$\begin{array}{l@{\null}l}&t' = r\\ \land& g' = n - 2\\\land&n'=n\end{array}$} (d);
  \draw (d) edge[->] node[below ]{call \cinline{fib}} (e);
  \draw (e) edge[->] node[right]{$r' = r + t$} (x);

  \draw [dashed,->] (b) -- ($ (b) - (2.75cm,0) $) -- ($ (r) - (2.75cm,0) $) -- (r);
  \draw [dashed,->] (d) -- ($ (d) - (3cm,0) $) -- ($ (r) - (3cm,-0.1cm) $) -- ($ (r.west) - (0,-0.1cm) $);
  \end{tikzpicture}
  \end{minipage}
  \caption{The recursive Fibonacci function (top), and representation
    as an inter-procedural control flow graph (bottom). The parameter
    and return are represented by the global variables $g$ and $r$
    (respectively); $t$ is a local temporary variable used to store
    the return value of the first recursive call.  Dashed edges are
    inter-procedural. \label{fig:fib}}
\end{figure}
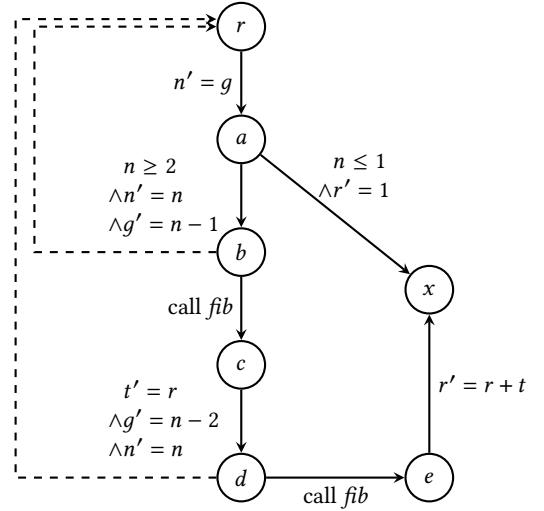

\section{Modular Design of Mortal Precondition Operators} \label{sec:mp-operator-design}

The interface provided by algebraic termination analysis is that the
analysis designer provides a mortal precondition operator for
transition formulas, and the framework ``lifts'' it to compute mortal
preconditions for whole programs.  Example~\ref{ex:LLRF} gives one
instantiation of the mortal precondition operator using linear
lexicographic ranking functions.  This section demonstrates the
applicability of the framework, by describing several combinators that
can be used to construct monotone mortal precondition operators.
A common theme to all is to take advantage of the properties of the
algebraic framework (compositionality and monotonicity, in
particular).

%The first combinator constructs a mortal precondition operator from
%an over-approximate transitive closure operator.
%Next we give a construction that exploits phase structure in loops to improve a given 
%mortal precondition 
%operator.
%Finally we discuss how mortal precondition operators can
%be combined. 

\subsection{Termination Analysis for Free} \label{sec:mp-exp}

Summarizing loops using an over-approximate transitive closure operator
is an integral component of our algebraic framework.
This section demonstrates that loop summarization can be also be
exploited to construct a mortal precondition operator; i.e., an
algebraic analyses for safety can be extended to prove termination analysis without any burden on the analysis designer.

Let $F$ be a transition formula.  A sufficient (but not necessary) condition for a state $s$ of $F$ to be mortal
is that there is a bound on the length of any execution starting from $s$; that is there is some
  $k$ such that for all $s'$ with $[s,s'] \models F^{k}$, $s'$
  has no $F$-successors.  This condition is not decidable, but it
  can be under-approximated using the procedure $\exp$
  described in Section~\ref{sec:transition-formulas} (or any other
  method for over-approximating the iterated behavior of a transition
  formula); this yields the following mortal precondition operator:
  \[ \mp_{\exp}(F) \defeq \exists k. \forall \Var',\Var''. k \geq 0 \land (\exp(F,k) \Rightarrow \lnot G)\]
  where $G \defeq F[\Var \mapsto \Var',\Var' \mapsto \Var'']$.

  The fact that $\mp_{\exp}$ is monotone follows from
  the monotonicity of quantification, conjunction, and the
  $\exp$ operator, and the fact that $F$ and
  $\exp(F,k)$ appear in negative positions in the formula.
 
\begin{example} \label{ex:mp-exp-example}
  Consider the loop
  \begin{center}
    \cinline{while (x != 0): x := x - 2},
  \end{center}
  with corresponding transition formula $F \defeq x \neq 0 \land x' = x - 2$.  In this case, we have $\exp(F,k) \models x' = x - 2k$, and
$\mp_{\exp}$ computes the exact precondition for termination of the loop:
\begin{align*}
  \mp_{\exp}(F) &\equiv \exists k.\forall x',x''. k \geq 0 \\
  & \quad  \land (x' = x - 2k \Rightarrow \lnot (x' \neq 0 \land x'' = x' - 2))\\
  &\equiv \exists k. k \geq 0 \land x-2k = 0
\end{align*}
i.e., the loop terminates provided that it begins in a state where
$x$ is a non-negative even number.
\end{example}

\subsection{Phase Analysis} \label{sec:mp-phase-analysis}
  This section describes a \textit{phase analysis} combinator that
  improves the precision of a given mortal precondition operator.
  The idea is to extract a \textit{phase transition graph} from  a transition formula, in which each vertex represents a phase of a loop, and each edge represents a phase transition.  Using
  the algebraic framework from Section~\ref{sec:ata} and a given mortal precondition operator $\mp$, we compute a mortal precondition for the phase transition graph, which (under mild assumptions) is guaranteed to be weaker than applying $\mp$ to the original transition formula (see Theorem~\ref{thm:phase-usefulness}).
  An important feature of phase analysis is that it can address the
  challenge of generating conditional termination arguments: even if some phases do not terminate, we can still use phase analysis to
  synthesize non-trivial mortal preconditions.

  %The possible interaction between phases might be expressed as a 
  %\textit{phase transition graph}, where paths in this graph 
  %correspond to possible sequences of phases permitted by the loop.
  %This setup, which involves graphs and paths, reminds us of the 
  %algebraic termination analysis framework (Section~\ref{sec:ata}) again.
  %In this section, we demonstrate the versatility of the 
  %algebraic termination framework that it can even be used to construct
  %mortal precondition operators (to be used in \textit{another} algebraic termination framework)!
  
  Let $F$ be a transition formula.  We say that a transition formula
  $p$ is \textbf{$F$-invariant} if, should some transition of $F$
  satisfy $p$, then so too must any subsequent transition; that is,
  the formula $(F \land p) \circ (F \land \lnot p)$ is inconsistent.
  Let $P$ be a fixed set of transition formulas (e.g., in our
  implementation, we take $P$ to be the set of all direction
  predicates, $P = \set{ x \bowtie x' : x \in \Var,
    \bowtie \in \{<,=,>\} }$).  Let $I(F,P)$ denote the $F$-invariant
  subset of $P$; $I(F,P)$ can be computed by checking the invariance
  condition for each formula in $P$ using an SMT solver.  The set of
  predicates $I(F,P)$ defines a partition $\mathcal{P}(F,P)$ of the
  set of transitions of $F$, where each cell corresponds to a
  valuation of the predicates in $P$ (i.e., each cell has the form \[F
  \land \left(\bigwedge_{p \in X} p\right) \land \left(\bigwedge_{p
    \in I(F,P)\setminus X} \lnot p\right)\ ,\] where $X$ is a subset of
  $I(F,P)$).  Since the predicates in $I(F,P)$ are $F$-invariant, this
  partition has the property that any infinite computation of $F$ must
  eventually lie within a single cell of the partition.
    
 %   Before we start, we define an algebra of transition formulas
%   $\mathbf{A} = \tuple{A, 0^A, 1^A, +^A, \cdot^A, \null^{*^{A}}}$ such that 
%   all elements of this algebra are transition formulas over a fixed 
%   set of program variables, $0^A$ is the empty transition where no transition 
%   is allowed over any pair of program states, $1^A$ is the unit 
%   transition that keeps all program variables the same, 
%   $T_1 +^A T_2$ is defined as $T_1 \cup T_2$, $T_1 \cdot^A T_2$ is 
%   the relational composition $T_1 \circ T_2$, and finally 
%   $T^{*^A}=\exists k. T^k$.

  Define the \textbf{phase transition graph} $\textit{Phase}(F,P)$ to be a
  labeled control flow graph where the vertices are the cells of the partition
  $\mathcal{P}(F,P)$ plus a root vertex $s$, and which has the
  following properties: (1) each cell has a self-loop, labeled by the cell (2)
  if cell $F_j$ can immediately follow $F_i$ (i.e., $F_i \circ F_j$ is
  satisfiable), there is an edge from $F_i$ to $F_j$ with label
  $1^\TF$ (3) there is an edge from $s$ to every cell with label $1^\TF$. The
  idea is that any infinite sequence $s_0 \rightarrow_F s_1 \rightarrow_F
  \dotsi$ corresponds to an $\omega$-path starting from $s$ in $G$.  Observe
  that this property is maintained if we relax conditions (2) and (3) so that
  we require only $1^\TF$-labeled \textit{paths} rather than edges;  call a
  phase transition graph \textit{reduced} if it satisfies the relaxed
  conditions, and the number of edges is minimal. An algorithm that constructs a
  reduced phase transition graph is given in
  Algorithm~\ref{alg:reduced-phase-graph}.

\begin{algorithm}
  \caption{Phase transition graph construction \label{alg:reduced-phase-graph}}
  %% \Fn{$\textsf{can-follow}(F, G)$} {
  %%   \Input{Two transition formulas $F$ and $G$}
  %%   \Output{Whether transition $G$ can take place immediately following $F$}
  %%   \Return{$\textsf{check-SAT}(F \circ G)$}
  %% }
  \Fn{$\textsf{phase-transition-graph}(F,P)$}{
    \Input{Formula $F$, set of transition predicates $P$}
    \Output{Reduced phase transition graph for $F$ and $P$}
    \tcc{$S$ is the set of literals for $F$-invariant predicates in $P$}
    $S \gets I(F,P) \cup \{ \lnot p : p \in I(F,P) \}$\;
    \tcc{Compute the cells of $\mathcal{P}(F,P)$}
    $n \gets 0$\;
    \While{$F \land \bigwedge_{i=1}^n \lnot F_i$ is SAT}{
      Select a model $t$ with $t \models F \land \bigwedge_{i=1}^n \lnot F_i$\;
      $n \gets n + 1$\;
      $F_n \gets F \land \bigwedge \{ p \in S : t \models p \}$\;
    }
    \tcc{Compute phase transitions}
    Sort $F_1,\dots,F_n$ by \# of positive literals\;
    $E \gets \{\}$\;
%     $V \gets \textsf{sort}(V')$ where the sort uses the number of 
%     transition predicates that appear as positive literals in $v_i = F_i$ as the key\;
     \For{$i$ = $2$ to n} {
       \For{$j$ = $i-1$ downto $1$}{
         \If{$\tuple{F_j,F_i} \notin E^*$ and $F_j \circ F_i$ is SAT}{
%            \tcc{Add edge $(v_j, v_i)$ with label $1$ where $1$ is the unit in the 
%            algebra of transition formulas $\TF$}
           $E \gets E \cup \set{\tuple{F_j, F_i}}$\;
         }
       }
     }
     \tcc{Connect virtual start node $s$ to unreachable vertices}
     $E \gets E \cup \set{ \tuple{s,F_i} : \nexists j. \tuple{F_j,F_i} \in E }$\;
     $E \gets E \cup \set{ \tuple{F_i,F_i} : 1 \leq i \leq n }$ \tcc*{Add self-loops}
     $\elbl \gets \lambda (F_i,F_j). \textbf{if } i = j \textbf{ then } F_i \textbf{ else } 1^\TF$\;
     \Return{$\tuple{\tuple{\{s, F_1,\dots,F_n\},E,s},L}$}
  }
\end{algorithm}

  We now define the phase analysis combinator.
  Suppose that $\mp$ is a mortal precondition operator; define
  the mortal precondition operator
  $\mp_{\textit{Phase}(P,\textit{mp})}$ as follows.  Let $F$ be a transition formula.  Construct the (reduced) phase transition graph $\tuple{G = \tuple{V,E,s}, \elbl}$ using Algorithm~\ref{alg:reduced-phase-graph}.  Compute an $\omega$-path expression $\opathexp{G}{s}$ for $G$ as in Section~\ref{sec:pathexp}.
  Define an interpretation $\mathcal{T} \defeq \tuple{\TF,\MP,\elbl}$, where the $(-)^{\omega^\MP}$ operator is taken to be $\mp$.
  Finally, define
  \[
  \mp_{\textit{Phase}(P,\textit{mp})}(F) \defeq \mpsem{\opathexp{G}{s}} \ .
  \]

\begin{restatable}[Soundness]{theorem}{PhaseSoundness} \label{thm:phase-soundness}
  Let $\mp$ be a  mortal precondition operator and let $P$ be a set of transition predicates.   
  Then $\mp_{\textit{Phase}(P,\textit{mp})}$ is a mortal precondition operator.
\end{restatable}

\begin{restatable}[Guaranteed improvement]{theorem}{PhaseUsefulness} \label{thm:phase-usefulness}
    Let $\mp$ be a monotone mortal precondition operator and let $P$ be a set of transition predicates.   
    Suppose that for any transition formula $F$, we have $\wp(F^\star,\mp(F)) = \mp(F)$.
    Then $\mp(F) \models \mp_{\textit{Phase}(P,\textit{mp})}(F)$.
\end{restatable}

  \begin{restatable}[Monotonicity]{theorem}{PhaseMonotonicity} \label{thm:phase-monotonicity}
  Let $\mp$ be a monotone mortal precondition operator and let $P$ be a set of transition predicates.   Suppose that for any transition formula $F$, we have $\wp(F^\star,\mp(F)) = \mp(F)$.
  Then the mortal precondition operator $\mp_{\textit{Phase}(P,\textit{mp})}$ is monotone.
  \end{restatable}

  \begin{example} \label{ex:phase}
    Consider the loop in Figure~\ref{fig:phase}.  
    The loop does not always terminate, so $\mp_{\textit{LLRF}}$ (Example~\ref{ex:LLRF}) computes a trivial mortal precondition ($x \leq 0$). 
  However, Algorithm~\ref{alg:reduced-phase-graph} discovers a phase structure for this loop: 
  once it execute the \textbf{then} branch,
  it cannot ever execute the \textbf{else} branch;
  the opposite is also true.   Within the \textbf{then} branch,
  the variable $x$ may increase (or remain constant) for some transient period, but then must ultimately decrease.  This structure is depicted in the phase transition graph in Figure~\ref{fig:phase-ex-phase-graph}.

  Although the original loop has no linear lexicographic ranking
  function, the two phases in the \textbf{then} branch do: $-y$ is a
  ranking function for phase $a$ and $x$ is ranking function for phase
  $b$.  The \textbf{else} branch does not, and so $\mp_{\textit{LLRF}}$
  generates a mortal precondition $x \leq 0 \lor f \geq 0$ (which is
  the trivial mortal precondition for phase $m$, but is a precise
  description of the mortal states of the original loop).  Thus, by computing the mortal precondition of the loop using its phase graph rather than applying the mortal precondition operator to the loop itself, we get a
  weaker mortal precondition.
  \end{example}
  
  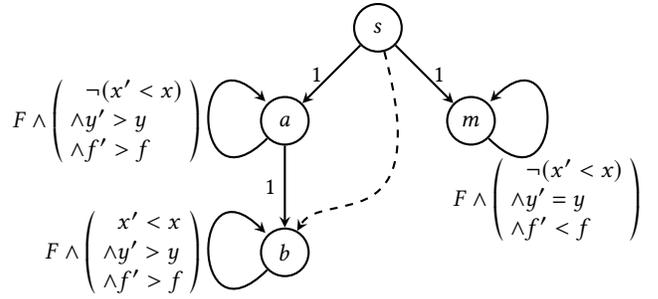
\begin{figure}
  \begin{minipage}{\linewidth}
  \begin{subfigure}{0.5\linewidth}
  \small
  \begin{lstlisting}[style=base,numbers=left,xleftmargin=2em]
while (x > 0):
  if ( f >= 0 ):
    x := x - y
    y := y + 1
    f := f + 1
  else
    x := x + 1
    f := f - 1\end{lstlisting}
  \caption{A loop with phase structure}
  \end{subfigure}
    \begin{subfigure}{0.45\linewidth}
    \small
    \[ \begin{array}{l@{}l}
      & x > 0\\
      \land & \tightparen{ \begin{array}{l@{}l}
      & \tightparen{ \begin{array}{l@{}l}
      & f \geq 0\\
       \land & x' = x - y\\
       \land & y' = y + 1\\
       \land & f' = f + 1\end{array}
      }\\
      \lor & \tightparen{ \begin{array}{l@{}l}
      & f < 0\\
      \land & x' = x + 1\\
      \land & f' = f - 1\\
      \land & y' = y \end{array}}
    \end{array}}
    \end{array} \]
      \caption{Loop transition formula, $F$}
    \end{subfigure}
  \end{minipage}
  \begin{subfigure}{\columnwidth}
    \centering
      \begin{tikzpicture}[thick,>=stealth,
        v/.style={circle,draw,minimum height=18pt},
        node distance=1.75cm
        ]
      \small
    \node [v] (s) {$s$};
    \node [v,below left of=s] (a) {$a$};
    \node [v,below right of=s] (m) {$m$};
    \node [v,below of=a] (b) {$b$};
    \draw (s) edge[->] node[left]{$1$} (a);
    \path[->] (a) edge [in=135,out=-135,looseness=3,loop] node[left] {$F \land \tightparen{\begin{array}{l@{\null}l} 
       & \lnot (x' < x)\\
      \land & y' > y \\  
      \land & f' > f
      \end{array}}$}
    (); 
    \draw (a) edge[->] node[left]{$1$} (b);
    \path[->] (b) edge [in=135,out=-135,looseness=2,loop] node[left] {$F \land \tightparen{\begin{array}{l@{\null}l} 
       & x' < x\\
      \land & y' > y \\  
      \land &  f' > f
      \end{array}}$} 
    (); 
    \draw (s) edge[->] node[right]{$1$} (m);
    \path[->] (m) edge [in=45,out=-45,looseness=2,loop] node[below=10pt] {$F \land \tightparen{\begin{array}{l@{\null}l} 
       & \lnot (x' < x)\\
      \land & y' = y\\ 
      \land & f' < f
      \end{array}}$} 
    (); 
    \draw[dashed,->] (s) .. controls +(-75:3cm) and +(60:1cm) .. (b);
    \end{tikzpicture}
    \caption{Phase transition graph for $F$.
    Solid edges form a reduced phase transition graph. 
     \label{fig:phase-ex-phase-graph}
     }
  \end{subfigure}
%  \begin{subfigure}{\columnwidth}
%    \centering
%      \begin{align*}
%        &\overbrace{\tuple{s,a}\tuple{a,a}^\omega + \tuple{s,a}\tuple{a,a}^*\tuple{b,b}^\omega}^{\text{then branch}} \\
%      + &\underbrace{\tuple{s,m}(\tuple{m,m})^\omega}_{\text{else branch}} 
%      \end{align*}
%  \caption{$\omega$-path expression starting from virtual start node $s$. 
%  \label{fig:phase-ex-pe}}
%  \end{subfigure}
  \caption{Analysis of a loop with a phase structure
   \label{fig:phase}}
  \end{figure}

\subsection{Combining Mortal Precondition Operators}

State-of-the-art termination analyzers often use a portfolio of techniques to prove termination.  Heuristics for selecting among appropriate techniques in a portfolio can be another source of unpredictable (non-monotone) behavior.  A feature of our framework is that it makes it easy to combine the strengths of different termination analyses without such heuristics.

Suppose that $\mp_1$ and $\mp_2$
are mortal precondition operators.  Then we can combine $\mp_1$ and
$\mp_2$ into a single mortal precondition operator $\mp_1 \otimes \mp_2$ by defining
\[ (\mp_1 \otimes \mp_2)(F) \defeq \mp_1(F) \lor \mp_2(F)\ ;\]
if $\mp_1,\mp_2$ are
monotone, then so too is $\mp_1 \otimes mp_2$.

In fact, monotonicity allows us to do better.  Define a second combinator by
\[ (\mp_1 \ltimes \mp_2)(F) \defeq \mp_2(F \land \lnot \mp_1(F))\ .\]
The intuition is that $\mp_1 \ltimes \mp_2$ is an \textit{ordered} product, which asks $\mp_2$ only to find a mortal precondition for the region of the state space that $\mp_1$ cannot prove to be mortal.  If we suppose that for all $F$ we have $\textit{Pre}(F) \models \mp_2(F)$, then we have (for all $F$)
\[ (\mp_1 \otimes \mp_2)(F) \models (\mp_1 \ltimes \mp_2)(F) \ .\]

\section{Evaluation} \label{sec:evaluation}

Our tool ComPACT (\textbf{Com}positional and \textbf{P}redictable
\textbf{A}nalysis for \textbf{C}onditional \textbf{T}ermination)
implements the algebraic program analysis framework described in Sections~\ref{sec:pathexp} and~\ref{sec:ata}), two mortal precondition operators
$\mp_{\textit{LLRF}}$ (Example~\ref{ex:LLRF}) and $\mp_{\exp}$
(Section~\ref{sec:mp-exp}), and the combinator $\mp_{\textit{Phase}}$
(Section~\ref{sec:mp-phase-analysis}). ComPACT's default mortal precondition operator is
$\mp_{\textit{Phase}(P,\mp_{\textit{LLRF}} \ltimes \mp_{\exp})}$ (where $P$ is a set of direction predicates, $P \defeq \set{ x \bowtie x' : x \in \Var,
    \bowtie \in \{<,=,>\} }$).
 We compare ComPACT
against Ultimate Automizer~\cite{TACAS:DietschHNSS20},
2LS~\cite{TOPLAS:CDKSW2018}, and CPAchecker~\cite{CPACheckerTermination}, the top
three placing competitors in the termination category of the Competition on Software Verification (SV-COMP) 2020\footnote{\href{https://sv-comp.sosy-lab.org/2020}{https://sv-comp.sosy-lab.org/2020}}.  We also
compare with Termite \cite{PLDI:GMR2015}, which implements a complete procedure
for linear lexicographic ranking function (LLRF) synthesis, to
evaluate the effectiveness of our algebraic framework. With the exception of
2LS, all tools treat variables as unbounded integers.

\begin{table*}
\centering
\caption{Termination verification benchmarks; time in seconds.     \label{tab:all-tools-all-benchmarks}}
\begin{tabular}{@{}lc|c@{}r|c@{}r|c@{}r|c@{}r|c@{}r@{}}
\toprule
 & & \multicolumn{2}{c|}{ComPACT} & \multicolumn{2}{c|}{2LS} & \multicolumn{2}{c|}{UAutomizer} & \multicolumn{2}{c|}{CPAchecker} & \multicolumn{2}{c}{Termite}\\
 benchmark & \#tasks & \#correct & time & \#correct & time & \#correct & time & \#correct & time & \#correct & time\\\midrule
termination & 171 & 141 & \textbf{81.7} & 115 & 1925.8 & \textbf{161} & 4684.8 & 126 & 13434.6 & 78 & 937.5\\
bitprecise & 169 & 115 & \textbf{154.3} & 111 & 1911.8 & \textbf{122} & 26596.5 & 92 & 32755.8 & 4 & 693.8\\
recursive & 42 & \textbf{31} & \textbf{49.6} & -- & -- & 30 & 2073.8 & 23 & 710.1 & -- & --\\
polybench & 30 & \textbf{30} & 93.8 & 0 & 7944.2 & 0 & 16285.8 & 0 & 4397.3 & 26 & \textbf{36.7}\\
\midrule
Total & 412  & \textbf{317} & \textbf{379.5} & 226 & 11781.8 & 313 & 49640.8 & 241 & 51297.8 & 108 & 1668.0\\
\bottomrule
\end{tabular}
\end{table*}

\begin{table*}
\centering
\caption{Contributions of different components implemented in ComPACT; time in seconds.     \label{tab:compact-ablation}}
\begin{tabular}{@{}lc|c@{}r|c@{}r|c@{}r|c@{}r|c@{}r@{}}
\toprule
 & & \multicolumn{2}{c|}{\multirow{2}{*}{ComPACT}} & \multicolumn{4}{c|}{Using $\mp_{\textit{LLRF}}$ as base operator } & \multicolumn{4}{c}{Using $\mp_{\exp}$ as base operator }  \\
& & \multicolumn{2}{c|}{} & \multicolumn{2}{c|}{LLRF only} & \multicolumn{2}{c|}{LLRF + phase}    & \multicolumn{2}{c|}{exp only} & \multicolumn{2}{c}{exp + phase}     \\
benchmark &  \#tasks & \#correct & time & \#correct & time & \#correct & time & \#correct & time & \#correct & time\\\midrule
termination & 171 & \textbf{141} & 81.7 & 122 & \textbf{65.2} & 138 & 70.3 & 112 & 72.2 & 130 & 92.4\\
bitprecise & 169 & \textbf{115} & 154.3 & 105 & \textbf{134.5} & \textbf{115} & 142.6 & 103 & 175.7 & 113 & 240.6\\
recursive & 42 & \textbf{31} & 49.6 & 15 & \textbf{31.2} & 22 & 38.7 & 24 & 44.7 & \textbf{31} & 78.9\\
polybench & 30 & \textbf{30} & 93.8 & \textbf{30} & \textbf{60.8} & \textbf{30} & 93.4 & \textbf{30} & 86.5 & \textbf{30} & 565.0\\
\midrule
Total & 412  & \textbf{317} & 379.5 & 272 & \textbf{291.7} & 305 & 345.0 & 269 & 379.1 & 304 & 976.9\\
\bottomrule
\end{tabular}
   
\end{table*}

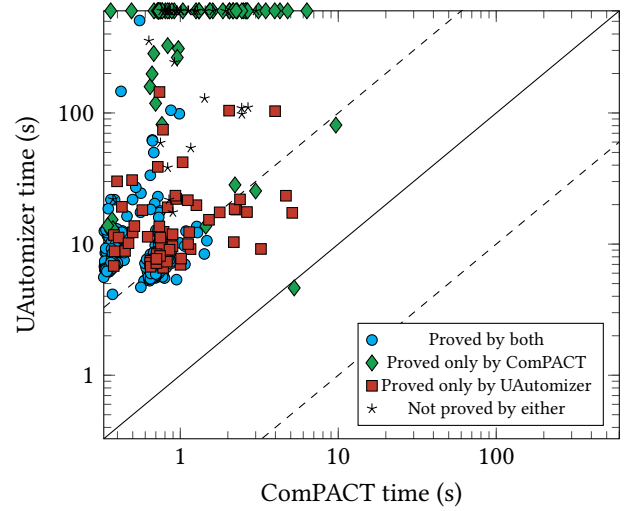
\begin{figure}
  \begin{tikzpicture}[scale=1]
  \begin{axis}[
      xlabel={ComPACT time (s)},
      ylabel={UAutomizer time (s)},
      xmax=601.0363106550003,
      ymax=601.0363106550003,
      xmin=0.3278084089999993,
      ymin=0.3278084089999993,
      xmode=log,
      ymode=log,
      legend pos=south east,
      mark size=2pt,
      log ticks with fixed point,
      legend style={nodes={scale=0.7, transform shape}}
    ]
    \addplot [only marks, style={solid, fill=cyan}, mark=*] table {scatter_ComPACT_uautomizer_tt.dat};
    \addlegendentry{Proved by both};
    \addplot [only marks, style={solid, fill=green}, mark=diamond*, mark size=3pt] table {scatter_ComPACT_uautomizer_tf.dat};
    \addlegendentry{Proved only by ComPACT};
    \addplot [only marks, style={solid, fill=red}, mark=square*] table {scatter_ComPACT_uautomizer_ft.dat};
    \addlegendentry{Proved only by UAutomizer};
    \addplot [only marks, style={solid}, mark=star ] table {scatter_ComPACT_uautomizer_ff.dat};
    \addlegendentry{Not proved by either};

    \addplot[mark=none,domain={0.3278084089999993:601.0363106550003}] {x};
    \addplot[mark=none,domain={0.3278084089999993:601.0363106550003},dashed] {10*x};
    \addplot[mark=none,domain={0.3278084089999993:601.0363106550003},dashed] {0.1*x};
  \end{axis}
\end{tikzpicture}
\caption{ComPACT vs. UAutomizer performance}
\label{fig:scatter-plot-performance}
\end{figure}

\paragraph{Environment} We ran
all experiments in a virtual machine with Ubuntu 18.04
and kernel version $5.3.0-62$, with a single-core Intel Core i7-10710U CPU @ 1.10GHz and 8GB of RAM.
All tools were run with a time limit of $10$ minutes. 

\paragraph{Benchmark design} 
We tested on a suite of 413 programs divided into 4 categories.
The \texttt{termination}, \texttt{bitprecise}, and \texttt{recursive} suites contain small programs with challenging termination arguments, while the \texttt{polybench}\footnote{\href{http://web.cs.ucla.edu/~pouchet/software/polybench}{http://web.cs.ucla.edu/{\textasciitilde}pouchet/software/polybench}} suite contains moderately sized kernels for numerical algorithms which have relatively simple termination arguments.
The \texttt{termination} category consists of the \textit{non-recursive, terminating} tasks in the \texttt{Termination-MainControlFlow} suite from SV-COMP.  The \texttt{recursive} category consists of the \textit{recursive, terminating} tasks from the \texttt{recursive} directory and \texttt{Termination-MainControlFlow}.      The \texttt{bitprecise} category consists of the same tasks as the \texttt{termination} category, except that bounded integer semantics 
are encoded into unbounded integer semantics for a more accurate comparison with 2LS (minus one task, which Ultimate Automizer was able to prove to be non-terminating).  Since signed overflow is undefined in C, proving termination necessitates proving absence of signed overflow.  The encoding was
 performed by using \texttt{goto-instrument} \cite{ESP:AKNT2013,CAV:AKT2013} to instrument the code with checks for signed overflow that enter an infinite loop on failure.

\paragraph{How does ComPACT compare with the state-of-the-art?}
A comparison of all tools across all test suites is shown in
Table~\ref{tab:all-tools-all-benchmarks}. Ultimate Automizer proves the most tasks in the \texttt{termination} and \texttt{bitprecise} suites, but
uses significantly more time than ComPACT (Figure~\ref{fig:scatter-plot-performance}).
ComPACT proves termination of the most tasks in the
\texttt{recursive} and \texttt{polybench} suite (note
that 2LS and Termite do not handle recursive programs, so we exclude them from
the \texttt{recursive} suite).
These results suggest that ComPACT is able to match or even exceed
the capabilities of state-of-the-art termination provers while
providing stronger behavioral guarantees.

%longer running times with a larger variance in general, even a few time outs.

\paragraph{How does each component contribute to ComPACT's capability?}
ComPACT implements 
two mortal precondition operators, LLRF-based mortal precondition operator
$\mp_{\textit{LLRF}}$ (LLRF) and transitive closure based $\mp_{\exp}$ (exp),
and the phase analysis (phase) combinator. 
We evaluate how each component contributes to ComPACT's ability
to prove termination in Table~\ref{tab:compact-ablation}.
First we notice that there is a large overlap between the tasks solved by 
$\mp_{\textit{LLRF}}$ and $\mp_{\exp}$.
This can be attributed to the fact that both are sufficient to prove termination of loops with linear ranking functions,
which is the case for the majority of the tasks in our suite.\footnote{We confirmed this fact by running the experiments with a modified $\mp_{\textit{LLRF}}$ that finds only linear ranking functions: it succeeds on 258 tasks without phase analysis and 292 tasks with phase analysis.}
The relative strength of $\mp_{\textit{LLRF}}$ is on loops with complex control structure (e.g., loops whose termination relies on precise reasoning about multiple paths through its body); the relative strength of $\mp_{\exp}$ is on loops with non-convex guards (e.g., recursive functions where the recursive case is guarded by a disequality).  
Theorem~\ref{thm:phase-usefulness} implies that the set of tasks that can be proved with phase analysis is a super-set of those that can be proved without; our experimental results show that the inclusion is strict for both configurations.
These results above suggest that the algebraic framework can be successfully applied using a variety of different mortal precondition operators, and that different operators can be profitably combined.

\paragraph{Impact of compositionality and monotonicity}
  The algebraic framework ``lifts'' a termination analysis for
  transition formulas to whole programs.  Comparing the LLRF column of Table~\ref{tab:compact-ablation} with the Termite results in
  Table~\ref{tab:all-tools-all-benchmarks} demonstrates the impact of this framework: both columns implement the same base analysis, but lift the analysis to whole programs in different ways.
  This comparison demonstrates the advantage of compositional summarization of nested loops,
  and also suggests that precision loss due to compositionality, i.e., synthesizing LLRFs
  without precise supporting invariants, is not substantial.

  A consequence of compositionality is that ComPACT has relatively
  stable running time across all tasks and scales to the larger tasks in the
  \texttt{polybench} suite.
  This suite contains program with loops
  that have complex control flow (e.g., nested loops) but
  simple termination arguments, in particular, \cinline{for} loops like
  \begin{lstlisting}[style=base,numbers=left,xleftmargin=2em]
  for(int i = 0; i < n; i++) { ... }\end{lstlisting}
  where the loop body does not contain instructions that decrease
  \cinline{i}.  ComPACT is assured to prove termination of such loops as
  a consequence of compositionality and
  monotonicity.  The other tools on our comparison, even those that employ
  complete procedures for linear ranking function synthesis, do not make 
  such guarantees and may get stuck in the logic of the loop body.  For
  example, ComPACT proves termination of the following loop in 0.3
  seconds:
  \begin{lstlisting}[style=base,numbers=left,xleftmargin=2em]
  for(int i = 0; i < 4096; i++)
    for(int j = 0; j < 4096; j++)
      i = i;\end{lstlisting}
  Ultimate Automizer and CPAchecker exceed the 10 minute time limit on
  this loop, and 2LS and Termite return ``unknown''.

\section{Related Work} \label{sec:related}

\paragraph{Summarization for termination}

At a high level, our procedure proves that a loop terminates by first
computing a transition formula that summarizes the behavior of its
body, and then performing termination analysis on the transition
formula.  There are several approaches to termination analysis that
similarly apply summarization to handle nested loops and procedure
calls
\cite{POPL:BCCDO2007,SAS:ZGSV2011,TACAS:TSWK2011,TOPLAS:CDKSW2018}.
There are various ways of formulating such an analysis.
\citet{POPL:BCCDO2007} generates loop body summaries using a program
transformation and a conventional state-based invariant generator
(e.g., polyhedra analysis). \citet{TACAS:TSWK2011} takes an approach
more similar to ours: summarization is an operation that replaces a
subgraph of a control flow graph by edges that summarize that
subgraph, and it is applied recursively to summarize nested loops.  We
take an algebraic view, inspired by
\citet{JACM:Tarjan1981,JACM:Tarjan1981b}, in which we generate an
$\omega$-regular expression representing the paths through a program
and then define the analysis by recursion on that expression.

The contribution of Section~\ref{sec:ata} is to provide a unified
framework in which these analyses can be understood.
In view of the algebraic framework, prior work can be understood in
terms of (1) the method used to summarize loops (i.e, the $(-)^*$
operator), and (2) the method used to prove termination (i.e, the
$(-)^\omega$ operator).
%\begin{itemize}
%\item \citet{SAS:ZGSV2011} summarizes inner loops
%and proves termination using a size-change abstraction
%\cite{POPL:LJB2001}.
%\item \citet{TACAS:TSWK2011} summarizes loops by using
%an SMT solver to check which of a given finite set of candidate
%transitive relations are implied by the loop body, and proves
%termination using a bit-vector ranking function synthesis algorithm
%\cite{TACAS:CKRW2010}.
%\item \citet{TOPLAS:CDKSW2018} (implemented in 2LS) summarizes inner
%  loops by havocing all possibly-changed variables, and proves
%  termination using a template-based method for synthesizing
%  bit-precise linear lexicographic ranking functions.
%\item \citet{POPL:BCCDO2007} uses iterative abstract interpretation
%  using various abstract domains (octagons, polyhedra, separation
%  logic) to compute summaries, and uses {\sc PolyRank}
%  \cite{ICALP:BMS2005,VMCAI:BMS2005} and {\sc RankFinder}
%  \cite{VMCAI:PR2004} to prove termination.
%\end{itemize}
A concrete benefit of our framework in light of this prior work is
that our approach handles recursive procedures and irreducible control
flow, which are not supported by some of the prior approaches (including 2LS)
\cite{SAS:ZGSV2011,TACAS:TSWK2011,TOPLAS:CDKSW2018}.

\paragraph{Complete ranking function synthesis}

A ranking function synthesis algorithm is \textit{complete} if it is
guaranteed to find a ranking function for a loop if one exists.  Such
techniques are related to our work in that we sought a termination
analysis for which we can make guarantees about its behavior.
Complete ranking function synthesis algorithms exist for a variety of
classes of ranking functions, such as linear \cite{VMCAI:PR2004},
linear-lexicographic \cite{CAV:BMS2005}, nested \cite{TACAS:LH2014},
multi-phase \cite{CAV:BG2017}, \dots).  These algorithms apply only to
very restricted classes of loops, and in particular there are no
complete ranking function synthesis algorithms that operate on nested
loops or recursive procedures.  The seminal work on {\sc Terminator}
gives a general method for applying complete ranking function
synthesis algorithms to general programs by using them in a
counter-example guided refinement loop \cite{PLDI:CPR2006}.  Our
framework of \textit{algebraic termination analysis} provides another
general method, which allows the \textit{completeness} guarantee
 to carry over to a \textit{monotonicity}
guarantee for the whole analysis.

\paragraph{Conditional termination}

In a compositional setting it is natural to formulate the termination
problem as the problem of finding a sufficient condition under which a
fragment of code is guaranteed to terminate (i.e., a \textit{mortal
  precondition}), rather than the decision problem of universal
termination.  
Approaches to conditional termination include quantifier
elimination \cite{CAV:CGLRS2008}, abstract interpretation
\cite{POPL:CC2012,SAS:Urban2013,ESOP:UM2014,SAS:UM2014}, abductive
inference \cite{PLDI:LQC2015}, conflict-driven learning
\cite{CAV:DU2015}, incremental backwards reasoning \cite{CAV:GG2013}, and
constraint-based methods \cite{TACAS:BBLORR2017}.  Our approach is
unique in that we provide a conditional termination analysis that is both
monotone and can be applied to a general program model.

\citet{TACAS:BIK2012} is closest to our work in that they give an
algorithm for which there are guarantees about its behavior beyond
soundness, albeit for a limited class of loops.  They give a technique
for synthesizing the set of mortal states of a loop,
provided a logical formula representing the exact transitive closure
of that loop in a logical theory that admits quantifier elimination.
In Section~\ref{sec:mp-exp}, we use a related idea to \textit{under-approximate} the mortal states of a loop using an
\textit{over-approximation} of the transitive closure of the loop.
%% \citet{TACAS:BIK2012} additionally give a specialized procedure for
%% the case that the loop is expressed as a difference bound or octagonal
%% relation; this technique also yields a monotone mortal precondition
%% operator, by using optimization modulo theories
%% \cite{IJCAR:ST2012,POPL:LAKGC2014} to compute the best approximation
%% of the input transition formula as a difference bound (or octagonal)
%% relation and then synthesizing the conditions under which this
%% relation is well-founded.

\paragraph{Control flow refinement}

Section~\ref{sec:mp-phase-analysis} defines a mortal precondition
combinator that improves the precision of a given mortal precondition
operator by exposing phase structure in loops.  There are several
related approaches for improving analysis results by program
transformation
\cite{TOPLAS:RM07,EMSOFT:BSIG09,PLDI:GJK09,CAV:SDDA2011,APLAS:FH14,POPL:CBKR2019,CAV:FWSS2019}.
In particular, the \textit{transition invariant predicates} from
Section~\ref{sec:mp-phase-analysis} are essentially a
transition-predicate analogue of the (state-based) \textit{splitter
  predicates} from \cite{CAV:SDDA2011}; our method for checking
whether a candidate transition predicate is invariant and partitioning
the transition space are new.  Cyphert et al.'s work
\cite{POPL:CBKR2019} on refinement of path expressions is closest to
ours in that it is based on an algebraic program analysis and 
provides a guarantee of improvement.  The refinement strategy is based
on altering the path expression algorithm, whereas phase
analysis alters the algebra of the analysis.  Operating at the algebra
level enables us to formulate and prove a monotonicity theorem.

\section{Conclusion} \label{sec:conclusion}

This paper presents a termination analysis that is both
\textit{compositional} and \textit{monotone}.  We extended
\citet{JACM:Tarjan1981,JACM:Tarjan1981b}'s path expression method from
safety analysis to termination analysis, by using $\omega$-regular
expressions to represent languages of infinite paths and
$\omega$-algebras to interpret those expressions.  One direction for
future work is to apply this framework to other analyses that require
reasoning about infinite and potentially infinite paths, such as
non-termination analysis, resource bound analysis, and verification of
linear temporal properties.

\begin{acks}
This work was supported in part
by the NSF under grant number 1942537 and
by ONR under grant N00014-19-1-2318.
Opinions, findings, conclusions, or recommendations
expressed herein are those of the authors and do not necessarily
reflect the views of the sponsoring agencies.
\end{acks}

\bibliography{references}
% -*- mode: LaTeX; -*-
\appendix

\section{Proofs}

First, we observe that $\MP$ and $\TF$ satisfy the following algebraic
laws (where formulas are considered to be equal if they are logically
equivalent):
\begin{lemma} \label{lem:idempotent-semiring}
  $\tuple{\TF, +^\TF, \cdot^\TF, 0^\TF, 1^\TF}$ is an idempotent semiring:
  \begin{itemize}
  \item $\cdot^\TF$ is associative and has $1^\TF$ as its identity
  \item $+^\TF$ is associative, commutative, idempotent, and has $0^\TF$ as its identity
  \item $\cdot^\TF$ distributes over $+^\TF$ (on the left and right)
  \item $F \cdot^\TF 0^\TF = 0^\TF \cdot^\TF F = 0^\TF$ for any $F$
  \end{itemize}
\end{lemma}

\begin{lemma} \label{lem:module}
 $\MP$ is a module over $\TF$:
  \begin{itemize}
  \item $+^\MP$ is associative, commutative, and idempotent
  \item $F .^\MP (S_1 +^\MP S_2) = (F .^\MP S_1) +^\MP (F .^\MP S_2)$
  \item $(F_1 +^\TF F_2) .^\MP S = (F_1 .^\MP S) +^\MP (F_2 .^\MP S)$
  \item $(F_1 \cdot^\TF F_2) .^\MP S = F_1 .^\MP (F_2 .^\MP S)$
  \item $1^\TF .^\MP S = S$
  \end{itemize}
\end{lemma}

Since $+^\TF$ is associative, commutative, and idempotent, it defines
a partial order relation $\leq^\TF$, where $F \leq^\TF G$ iff $F +^\TF
G = G$.  Observe that since $+^\TF$ is disjunction, $\leq^\TF$
coincides with logical entailment.  Similarly, $+^\MP$ defines a
partial order $\leq^\MP$, which coincides with \textit{reverse}
logical entailment.  From Lemmas~\ref{lem:idempotent-semiring} and
Lemma~\ref{lem:module}, we see that the operations $+^\TF$,
$\cdot^\TF$, $+^\MP$, and $.^\MP$ are monotone with respect to
these orders.  We will show monotonicity of $.^\MP$; the other
operators are similar.  Suppose that $F_1 \leq^\TF F_2$ and $S_1
\leq^\MP S_2$--we wish to show that $F_1 .^\MP S_1 \leq^\MP F_2
.^\MP S_2$:
\begin{align*}
  F_2 .^\MP S_2 &= (F_1 +^\TF F_2) .^\MP S_2 & \text{Since } F_1 \leq^\TF F_2\\
  &= (F_1 .^\MP S_2) +^\MP (F_2 .^\MP S_2) & Lemma~\ref{lem:module}\\
  &\geq^\MP F_1 .^\MP S_2\\
  &= F_1 .^\MP (S_1 +^\MP S_2) & \text{Since } S_1 \leq^\MP S_2 \\
  &= (F_1 .^\MP S_1) +^\MP (F_1 .^\MP S_2) & Lemma~\ref{lem:module}\\
  &\geq^\MP F_1 .^\MP S_1
\end{align*}

Let $P = \tuple{G,\elbl}$ be a labeled control flow graph, with $G =
\tuple{V,E,r}$.  Define a computation to be a sequence $\tau =
\tuple{v_0,s_0}\tuple{v_1,s_1}\dots\tuple{v_n,s_n} \in (V \times
\textsf{State})^*$ such that for all $i$ we have $\tuple{v_i,v_{i+1}}
\in E$ and $[s_i,s_{i+1}] \models \elbl(v_i,v_{i+1})$.  Define an
$\omega$-computation to be an infinite sequence in $(V \times
\textsf{State})^\omega$ such that every finite prefix is a
computation.  For a computation $\tau =
\tuple{v_0,s_0}\tuple{v_1,s_1}\dots\tuple{v_n,s_n}$, define its
underlying path to be $\textit{path}(\tau) \defeq
\tuple{v_0,v_1}\tuple{v_1,v_2}\dots\tuple{v_{n-1},v_n}$; define
$\textit{path}^\omega$ analogously for $\omega$-computations.

\begin{lemma} \label{lem:sound-pathex}
  Let $P$ be a labeled control flow graph, let $\tau =
  \tuple{v_0,s_0}\tuple{v_1,s_1}\dots\tuple{v_n,s_n}$ be a computation
  of $P$, and let $e \in \RegExp(E)$.  If $e$ recognizes
  $\textit{path}(\tau)$, then $[s_0,s_n] \models \tfsem{e}$.
\end{lemma}
\begin{proof}
  By induction on $e$.
  \begin{itemize}
  \item Case $e$ is $\tuple{u,v}$:  Since $e$ recognizes $\textit{path}(\tau)$, $\tau$ must take the form $\tuple{u,s}\tuple{v,s'}$ with $[s,s'] \models \elbl{u,v} = \tfsem{\tuple{u,v}}$.
  \item Case $e$ is $0$: trivial--$0$ does not recognize any paths.
  \item Case $e$ is $1$: $1$ recognizes only the empty path, so $\tau$ must have the form $\tuple{v_0,s_0}$, and $[s_0,s_0] \models 1^\TF = \tfsem{1}$.
  \item Case $e$ is $e_1e_2$: Since $e$ recognizes
    $\textit{path}(\tau)$, there is some $m$ such that $e_1$
    recognizes $\textit{path}(\tuple{v_0,s_0}\dots\tuple{v_m,s_m}))$
    and $e_2$ recognizes
    $\textit{path}(\tuple{v_m,s_m}\dots\tuple{v_n,s_n}))$.  By the
    induction hypothesis, $[s_0,s_m] \models \tfsem{e_1}$ and
    $[s_m,s_n] \models \tfsem{e_2}$.  It follows that $[s_0,s_n] \models \tfsem{e_1}\circ\tfsem{e_2} = \tfsem{e_1e_2}$.
  \item Case $e$ is $e_1^*$: Since $e$ recognizes
    $\textit{path}(\tau)$, there is some $i_0 = 0 ,i_1,\dots,i_m = n$
    such that $e_1$ recognizes the path
    $\textit{path}(\tuple{v_{i_j},s_{i_j}}\dots\tuple{v_{i_{j+1}},s_{i_{j+1}}})$
      for each $j$.  By the induction hypothesis, we have \[s_{i_0}
      \rightarrow_{\tfsem{e_1}} s_{i_1} \rightarrow_{\tfsem{e_1}}
      \dots \rightarrow_{\tfsem{e_1}} s_{i_n}\ ,\] and so $s_0
      \rightarrow_{\tfsem{e_1}}^* s_n$.  Since
      $\rightarrow_{\tfsem{e_1}}^* \subseteq
      \rightarrow_{\tfsem{e_1}^*}$, we have $[s_0,s_1] \models
      \tfsem{e_1}^\star = \tfsem{e_1^*}$.
  \end{itemize}
\end{proof}

\begin{lemma} \label{lem:sound-opathex}
  Let $P$ be a labeled control flow graph, let $\tau =
  \tuple{v_0,s_0}\tuple{v_1,s_1}\dots$ be an $\omega$-computation
  of $P$, and let $f \in \oRegExp(E)$.  If $f$ recognizes
  $\textit{path}^\omega(\tau)$, then $s_0 \not\models \mpsem{f}$.
\end{lemma}
\begin{proof}
  By induction on $f$.
  \begin{itemize}
  \item Case $f$ is $e^\omega$:  Since $f$ recognizes $\textit{path}^\omega(\tau)$, there is a sequence
    $0 = i_0, i_1, i_2, \dots$ such that $e$ recognizes the path $\tuple{v_{i_j},v_{i_j+1}}\dots\tuple{v_{i_{j+1}-1},v_{i_{j+1}}}$ for all $j$.  By Lemma~\ref{lem:sound-pathex}, we have
    $[s_{i_j},s_{i_{j+1}}] \models \tfsem{e}$ for all $j$.  It follows that
    \[ s_{i_0} \rightarrow_{\tfsem{e}} s_{i_1} \rightarrow_{\tfsem{e}} s_{i_2} \dotsi \]
    is an infinite computation in $\tfsem{e}$, and so $s_0 = s_{i_0}$
    is not a mortal state of $\tfsem{e}$.  Since $(-)^{\omega^\MP}$ is
    a mortal precondition operator, have  $s_0 \not\models \mpsem{e^\omega} = \tfsem{e}^{\omega^\MP}$
  \item Case $f$ is $f_1 + f_2$: Since $f$ recognizes $\textit{path}^\omega(\tau)$, we must have $f_1$ or $f_2$ recognize $\textit{path}^\omega(\tau)$.  Without loss of generality, suppose $f_1$ recognizes $\textit{path}^\omega(\tau)$.  By the induction hypothesis, $s_0 \not\models \mpsem{f_1}$, and therefore $s_0 \not\models \mpsem{f} = \mpsem{f}_1 \land \mpsem{f}_2$.
  \item Case $f$ is $e \cdot f'$: Since $f$ recognizes $\textit{path}^\omega(\tau)$, there is some $m$ such that $e$ recognizes $\tuple{v_0,v_1}\dots\tuple{v_{m_1},v_m}$ and $f$ recognizes
    $ \tuple{v_m,v_{m+1}}\tuple{v_{m+1},v_{m+2}}\dots $. \linebreak By Lemma~\ref{lem:sound-pathex}, we have
    $[s_0,s_m] \models \tfsem{e}$ and by the induction hypothesis we have $s_m \not\models \mpsem{f'}$.  It follows that $s_0 \not\models \mpsem{e \cdot f'} = \wp(\tfsem{e},\mpsem{f})$. \qedhere
  \end{itemize}
\end{proof}

\soundness*
\begin{proof}
  We show the contrapositive.  If $\tuple{r,s}$ is not mortal, then
  there is an $\omega$-computation $\tau$ of $P$ beginning with
  $\tuple{r,s}$.  Since $\opathexp{G}{r}$ recognizes all
  $\omega$-paths beginning at $r$, it must recognize
  $\textit{paths}^\omega(\tau)$.  By Lemma~\ref{lem:sound-opathex}
  we have $s \not\models \mpsem{\opathexp{G}{r}}$.
\end{proof}

\monotonicity*
\begin{proof}
  By induction on $f$.  The base case is immediate from the assumption
  that $L_1(\tuple{u,v}) \models L_2(\tuple{u,v})$ for all
  $\tuple{u,v} \in E$.  The inductive steps follow from monotonicity
  of all the operations of $\TF$ (Lemma~\ref{lem:idempotent-semiring},
  Lemma~\ref{lem:exp-monotone}) and $\MP$ (Lemma~\ref{lem:module} and
  the assumption that $\mp$ is monotone).
\end{proof}

\subsection{Phase analysis}

We first prove some properties of our 
over-approximating transitive closure operator $\star$.

% \begin{lemma}
%   For transition formula $T \in \TF$, if $1^\TF \subseteq T$, then for any state
%   formula $S$
%   \[
%     \wp(T, S) \models S
%   \]
%   \label{lem:wp-of-S-under-T-implies-S}
% \end{lemma}

%    F^\omega &= F^* .^\MP F^\omega \\

\begin{lemma}
  For any transition formula $F \in \TF$, state formula $S \in \MP$, we have
  \begin{align}
    1^\TF &\models F^\star\\
    F^\star &= F^\star \circ F^\star  \\
    F^\star .^\MP S &\models S
  \end{align}
  \label{lem:props-star-and-omega}
\end{lemma}

In the following, let $\textit{PTG}(F,P)$ denote the phase transition graph for the formula $F$ and set of predicates $P$.

\PhaseSoundness*
\begin{proof} 
Let $F$ be a transition formula, and let $G = \textit{PTG}(F,P)$.  Let $r$ be the root of $G$.
Towards the contrapositive,
suppose that there is an infinite $F$ computation
$s_0 \rightarrow_F s_1 \rightarrow_F s_2 \dotsi$.
For each pair of states $[s_i, s_{i+1}]$, there is a unique cell 
$c_i \in \mathcal{P}(F,P)$ such that
$[s_i,s_{i+1}]$ (since $\mathcal{P}(F,P)$ is a partition of the transitions of $F$).  Since for each $i$, we have  $s_i \rightarrow_F s_{i+1} \rightarrow_F s_{i+2}$, $[s_i,s_{i+1}] \models c_i$ and $[s_{i+1},s_{i+2}] \models c_{i+1}$, we have that $c_i \circ c_{i+1}$ is satisfiable, and so there is a $1^\TF$-labelled path from $c_i$ to $c_{i+1}$ in $G$.  Finally, there is a $1^\TF$-labelled path from $r$ to $c_0$.  It follows that there is an infinite execution of $\textit{TS}(G)$ starting from $\tuple{s_0,r}$, and so by Proposition~\ref{lem:ata-soundness} we have $s_0 \not\models \mpsem{\opathexp{G}{r}}$.
\end{proof}

Now we move on to prove the monotonicity of phase analysis.
First we define \textit{canonical form} of path expressions and 
prove some related results. 
\begin{definition}
An $\omega$-path expression $P$ is in canonical form if it has form 
\[
(P_{1}^* P_{2}^* \ldots P_{n}^* ) L^\omega
\]
where $P_i$'s and $L$ are letters.  
The transition formula part $P_{1}^*  P_{2}^* \ldots P_{n}^*$ is
called the \textit{stem} and the expression $L^\omega$ is called the
\textit{body}.
\end{definition}

Let $\mathcal{L}^\omega(f)$ denote the language recognized by the $\omega$-path
expression $f$.  For an $\omega$-path expression $f$ for a phase transition graph $\textit{PTG}(F,P)$, define $h_1(f)$ to be the $\omega$-regular expression obtained by deleting phase transition edges (recall: phase transition edges are labelled with $1^TF$--the loops of the phase transition graph contain all of its significant content).  Clearly, we have
$\mpsem{f} = \mpsem{h_1(f)}$.

%\begin{lemma}
%    For any $\omega$-regular expression $f$,
%    $\mpsem{h_1(f)} = \mpsem{f}$.
%\end{lemma}

\begin{lemma}[Canonical form of path expression]
  Let $F$ be a transition formula,  $P$ be a set of predicates,
  $G = \textit{PTG}(F,P)$, and $r$ be the root of $G$.
  There exists canonical form path expressions $P_1, \ldots, P_N$ such that 
  \[ 
  \mathcal{L}^\omega({h_1(\opathexp{G}{r})}) = \mathcal{L}^\omega(P_1 + P_2 + \ldots + P_N)
  \]
  Furthermore, we have
  \[
  \mpsem{\opathexp{G}{r}} = \mpsem{P_1 + P_2 + \ldots + P_N}] .
  \]
  \label{lem:path-exp-canonical-form}
\end{lemma}
\begin{proof}
  By structural induction on $\opathexp{G}{r}$ and Lemma~\ref{lem:idempotent-semiring} and \ref{lem:module}.
\end{proof}

\begin{lemma}
  Suppose that $P = c_1^* \dotsi c_k^* c_{k+1}^\omega$ is a canonical path
  expression and $\mathcal{L}^\omega(P) \subseteq \mathcal{L}^\omega(Q_1 + Q_2 + \ldots + Q_n) $
  where each $Q_i$ is also a canonical path. Then 
  there exists some $Q_i$ such that $\mathcal{L}^\omega(P) \subseteq \mathcal{L}^\omega(Q_i)$.
  \label{lem:canonical-path-inclusion}
\end{lemma}

\begin{lemma}
Suppose that $P = c_1^* \dotsi c_k^* c_{k+1}^\omega$ and $Q = d_1^* \dotsi
d_\ell^* d_{\ell+1}^\omega$ are canonical path expressions such that
$\mathcal{L}^\omega(P) \subseteq \mathcal{L}^\omega(Q)$.  Then there is a
monotone map $f : \{1,\dots,k+1\} \rightarrow \{1,\dots,\ell+1\}$ such that $f(k+1)
= \ell+1$ and $c_i = d_{f(i)}$ for all $i$.
\label{lem:canonical-path-mono-map-existence}
\end{lemma}

\begin{lemma}
Let $F$ be a transition formula, $P$ be a set of predicates, let $G = \textit{PTG}(F)$, and let $r$ be the root of $G$.  Let
$c_1, c_2, \ldots, c_k, c_{k+1}$ be a sequence of cells such that 
for every $i \in [k]$, $c_i \circ c_{i+1}$ is satisfiable.  
Suppose that for any transition formula $F$, $\wp(F^\star, \mp(F)) = \mp(F)$.
Then we have 
\[
\mpsem{\opathexp{G}{r}} \models \mpsem{(c_1^* \ldots c_k^*) c_{k+1}^\omega} 
\]
\label{lem:phase-mono-phase-path-lemma}
\end{lemma}
\begin{proof}
  We first show that 
  \[
    \mathcal{L}^\omega((c_1^* \ldots c_k^*) c_{k+1}^\omega) 
  \subseteq \mathcal{L}^\omega(h_1(\opathexp{G}{r}))
  \]
  Since all $c_i \circ c_{i+1}$ are satisfiable, we know that in the phase transition graph 
  there exists paths all with label $1^\TF$ from $c_i$ to $c_{i+1}$ for all $i$, and also from 
  the root of the phase transition graph to $c_1$.
  It follows that every string in $\mathcal{L}^\omega((c_1^* \ldots c_k^*) c_{k+1}^\omega)$
  is recognized by $h_1(\opathexp{G}{r})$.

  According to Lemma~\ref{lem:path-exp-canonical-form},
  there exist canonical path expressions
  $Q_1, \dots, Q_N$ such that 
  \[\mathcal{L}^\omega(h_1(\opathexp{G}{r})) = \mathcal{L}^\omega(Q_1 + \ldots Q_N)\ .\]
  By Lemma~\ref{lem:canonical-path-inclusion}, there exists $Q_i$ such that 
  \[
    \mathcal{L}^\omega((c_1^* \ldots c_k^*) c_{k+1}^\omega) \subseteq \mathcal{L}^\omega(Q_i)
  \]
  Let $Q_i = d_1^* \ldots d_\ell^* d_{\ell+1}^\omega$.
  By Lemma~\ref{lem:canonical-path-mono-map-existence}, 
  there is a monotone map $f : \{1,\dots,k+1\} \rightarrow \{1,\dots,\ell+1\}$ such that $f(k+1)
  = \ell+1$ and $c_i = d_{f(i)}$ for all $i$.
  Let $P = (c_1^* \ldots c_k^*) c_{k+1}^\omega$.

  We now prove that $\mpsem{Q_i} \models \mpsem{P}$ by induction.
  Specifically we prove that for all $1 \leq i \leq k + 1$,
  \[
  \mpsem{d_{f(i)}^* \ldots d_\ell^* d_{\ell+1}^\omega} \models 
   \mpsem{c_i^* \ldots c_k^* c_{k+1}^\omega}
   \]
  The base case, $i = k + 1$, is trivial.
  Since $c_{k+1} = d_{\ell+1}$, we have
  $\mpsem{d_{\ell+1}^\omega} \models \mpsem{c_{k+1}^\omega}$.
  
  Now we need to prove the induction step.
  Assuming the statement is true for $i$:
    \[
    \mpsem{d_{f(i)}^* \dotsi d_{\ell+1}^\omega} \models \mpsem{c_i^* \dotsi c_k^* c_{k+1}^\omega}
    \]
  We need to prove it for $i-1$:
    \[
    \mpsem{d_{f(i-1)}^* \dotsi d_{\ell+1}^\omega} \models \mpsem{c_{i-1}^* c_i^* \dotsi c_k^* c_{k+1}^\omega}
    \]
  There are $3$ separate cases to consider.
  \begin{enumerate}
  \item $f(i-1) = f(i) = f(k+1) = \ell+1$: In this scenario, 
  the inductive hypothesis is
  \[
  \mpsem{d_{\ell+1}^\omega} \models 
  \mpsem{c_i^*\dotsi c_k^* c_{k+1}^\omega}
  \]
  and we need to prove 
  \[
  \mpsem{d_{\ell+1}^\omega} \models 
  \mpsem{c_{i-1}^* c_i^*\dotsi  c_k^*c_{k+1}^\omega}
  \]
  By the fact that $d_{f(i-1)}=d_{\ell+1} = c_{i-1}$ and the assumptions of this lemma, we have
  \[
   \mpsem{d_{\ell+1}^\omega} = \mpsem{d_{\ell+1}^* d_{\ell+1}^\omega}
  \]
  thus by Lemma~\ref{lem:module}
  \begin{align*}
      \mpsem{d_{\ell+1}^\omega} =  d_{\ell+1}^* .^\MP \mpsem{d_{\ell+1}^\omega}
      &\models c_{i-1}^* .^\MP
  \mpsem{c_i^*\dotsi c_k^* c_{k+1}^\omega} \\ 
  &= \mpsem{c_{i-1}^* c_i^*\dotsi  c_k^*c_{k+1}^\omega}
  \end{align*}
  
 \item $f(i-1) = f(i) < f(k+1) = \ell+1$:
  Then $d_{f(i-1)} = d_{f(i)} = c_i = c_{i-1}$.
  According to Lemma~\ref{lem:props-star-and-omega} and inductive hypothesis,
    \begin{align*}
      \mpsem{d_{f(i-1)}^* \ldots d_\ell^* d_{\ell+1}^\omega} &=
      \mpsem{d_{f(i)}^* \ldots d_\ell^* d_{\ell+1}^\omega}  \\ 
      &\models \mpsem{c_i^* \ldots c_k^* c_{k+1}^\omega} \\
      &= \mpsem{c_{i-1}^* c_i^* \ldots c_k^* c_{k+1}^\omega}
    \end{align*}
  \item  $f(i-1)<f(i)$:
  By Lemma~\ref{lem:props-star-and-omega} and the inductive hypothesis,
\begin{align*}
 &\quad  \mpsem{d_{f(i-1)+1}^*  \ldots 
 d_{f(i)}^* \ldots d_\ell^* d_{\ell+1}^\omega} \\
 &\models \mpsem{d_{f(i)}^* \ldots d_\ell^* d_{\ell+1}^\omega} \\
 &\models \mpsem{c_i^* \ldots c_k^* c_{k+1}^\omega}
\end{align*}

By Lemma~\ref{lem:module} we have that since $d_{f(i-1)} = c_{i-1}$,
\begin{align*}
   &\quad d_{f(i-1)}^* .^\MP \mpsem{d_{f(i-1)+1}^*  \ldots  d_{f(i)}^* \ldots d_\ell^* d_{\ell+1}^\omega} \\
    &\models c_{i-1}^* .^\MP \mpsem{c_i^* \ldots c_k^* c_{k+1}^\omega}
\end{align*}
which implies
\begin{align*}
    &\quad \mpsem{d_{f(i-1)}^*  \ldots  d_{f(i)}^* \ldots d_\ell^* d_{\ell+1}^\omega} \\
    &\models \mpsem{c_{i-1}^* c_i^* \ldots c_k^* c_{k+1}^\omega}
\end{align*}
\end{enumerate}

\end{proof}

\PhaseUsefulness*
\begin{proof}
  Let $G = \textit{PTF}(F, P)$ and let $s$ be the root of $G$.
  We only need to prove that $\mp(F) \models \mpsem{\opathexp{G}{r}}$.
  Expanding the path expression on the RHS to its canonical forms according to
  Lemma~\ref{lem:path-exp-canonical-form}, we only need to show
  $\mp(F) \models \mpsem{P_1+P_2 \ldots + P_N}$,
  where the $P_i$'s are in canonical form.
  Since \[\mpsem{P_1+P_2 \ldots + P_N} =
  \mpsem{P_1} \land \dotsi \land \mpsem{P_N}\] it is sufficient to prove
  $\mp(F) \models \mpsem{P_i}$ for all $i$.
  Suppose that $P_i$ has form $(p_1^* p_2^* \dots p_n^*) q^\omega $,
  where $p_i$'s and $q$ are letters that correspond to edges of $G$.
  Suppose $L$ is the semantic function we used to define $\mpsem{*}$, 
  that maps each $p_i$ to a transition formula
  as labeled in $G$.
  Now consider another semantic function $L'$ that maps each $p_i$ to $F$
  and the correspondingly defined $\mathcal{T}_2 \defeq \tuple{\TF, \MP, L'}$.
  By construction of the phase transition graph, $L(p_i) \models L'(p_i)$ for all $i$
  and $L(q) \models L'(q)$.
  By Proposition~\ref{lem:ata-monotonicity}, we know 
  \[
  \mathcal{T}_2^\omega \!\sem{(p_1^* p_2^* \dots p_n^*) q^\omega} \models \mathcal{T}^\omega \!\sem{(p_1^* p_2^* \dots p_n^*) q^\omega}
  \]
  Thus it suffices to prove that $\mp(F) \models \mathcal{T}_2^\omega \!\sem{(p_1^* p_2^* \dots p_n^*) q^\omega}$,
  which is equivalent to $\mp(F) \models \mp((F^* F^* \dots F^*) F^\omega) $.
  This is obvious since we require $\wp(F^\star,\mp(F)) = \mp(F)$ in the theorem statement.
\end{proof}

\PhaseMonotonicity*
 \begin{proof}
   Let $F_1$ and $F_2$ be transition formulas with $F_1 \models F_2$.
   Let $G_1 = \textit{PTG}(F_1,P)$ and $G_2= \textit{PTG}(F_2,P)$, 
   and let $s_1$ and $s_2$ be the roots of $G_1$ and $G_2$, respectively.
   What we need to show here is
   \[
     \mpsem{\opathexp{G_2}{s_2}} \models \mpsem{\opathexp{G_1}{s_1}}
   \]

  Since $F_1 \models F_2$, we have that the $F_2$-invariant subset of $P$ is a
  subset of the $F_1$ invariant subset of $P$, and that $\mathcal{P}(F_1,P)$ is
  finer than $\mathcal{P}(F_2,P)$. For any cell $c \in \mathcal{P}(F_1,P)$,
  define $\textit{proj}(c)$ to be the unique cell of $\mathcal{P}(F_2,P)$ such
  that $c \models \textit{proj}(c)$.  Projection can be lifted to map
  $\omega$-regular expressions over $\mathcal{P}(F_1,P)$ to $\omega$-regular
  expressions over $\mathcal{P}(F_2,P)$ in the obvious way.  By monotonicity, we
  have $\mpsem{\textit{proj}(f)} \models \mpsem{f}$ for any $\omega$-regular
  expression $f$.

  Expanding the path expression on the RHS to its canonical forms according to
  Lemma~\ref{lem:path-exp-canonical-form}, we only need to show
  $\mpsem{\opathexp{G_2}{s_2}} \models \mpsem{P_1+P_2 \ldots + P_N}$, where the
  $P_i$'s are in canonical form.  Since \[\mpsem{P_1+P_2 \ldots + P_N} =
  \mpsem{P_1} \land \dotsi \land \mpsem{P_N}\] it is sufficient to prove
  $\mpsem{\opathexp{G_2}{s_2}} \models \mpsem{P_i}$ for all $i$.
  
  Since (by monotonicity) we have $\mpsem{\textit{proj}(P_i)} \models
  \mpsem{P_i}$, it is sufficient to show $\mpsem{\opathexp{G_2}{s_2}} \models
  \mpsem{\textit{proj}(P_i)}$.
  
  Write $P_i$ as $c_1^*c_2^* \dotsi c_k^*c_{k+1}^\omega$.  Then for each $j$, we
  have that $c_j \circ c_{j+1}$ is satisfiable (because there is a corresponding
  phase transition in $G_1$), and thus $\textit{proj}(c_j) \circ
  \textit{proj}(c_{j+1})$ is satisfiable.  It follows that the sequence
  $\textit{proj}(c_1),\textit{proj}(c_2),\dots,\textit{proj}(c_{k+1})$ meets the
  conditions of Lemma~\ref{lem:phase-mono-phase-path-lemma}, and so
  \[ \mpsem{\opathexp{G_2}{s_2}} \models \mpsem{\textit{proj}(P_i)}\ .\]

 \end{proof}

Lastly, we note that $\mp_{\text{exp}}$ and $\mp_{\text{LLRF}}$ satisfy the conditions of Theorem~\ref{thm:phase-monotonicity} (that is,
$\wp(F^\star,F^\omega) = F^\omega$ for any transition formula $F$).

\section{Interprocedural analysis}  \label{sec:interproc}

The algebraic framework extends to the interprocedural case
using the method of \citet{FMSD:CPR2009}.  This section provides a
sketch for how this extension works.  The essential point is that no
additional work is required on the part of the analysis designer to
extend an algebraic termination analysis to the interprocedural case:
the same analysis that is used to prove conditional termination for
loops also can be applied to prove conditional termination for
recursive functions, and the monotonicity results extend as well.

Suppose that the set of variables $\Var$ is divided into a set of
local variables $\LVar$ and a set of global variables $\GVar$.  A
program can be represented as a tuple \[P =
\tuple{V,E,\textit{Proc},\Lambda,\textit{entry},\textit{exit}}\ ,\] where
$\tuple{V,E}$ is a finite directed graph, $\textit{Proc}$ is a finite
set of procedure names, $\Lambda : E \rightarrow \TF \cup \textit{Proc}$ labels
each edge by either a transition formula or a procedure call, and
$\textit{entry},\textit{exit} : \textit{Proc} \rightarrow V$ are functions
associating each procedure name with an entry and an exit vertex.  Note
that procedures do not have parameters or return values, but these can
be modeled using global variables (see Figure~\ref{fig:fib} for an
example).

An \textit{activation record} is a pair $\tuple{v,s}$ consisting a
control flow vertex $v \in V$ and a state $s : \State \rightarrow
\mathbb{Z}$.  A \textit{stack} is a sequence of activation records;
let $\textsf{Stack}$ denote the set of stacks.  Define a transition
system $\textit{TS}(P) = \tuple{\textsf{Stack}, R(P)}$, where the
states are stacks, and where there is a transition
$\tr{\textit{TS}(P)}{\textit{stack}}{\textit{stack}'}$ iff one of the
three conditions hold:

\begin{itemize}
\item (Local) there is a transition
  \[\tr{\textit{TS}(P)}{\tuple{v,s}\textit{base}}{\tuple{v',s'}\textit{base}}\]
  for any stack $\textit{base}$, any activation record
  $\tuple{v,s}$, any vertex $v'$ and any state $s$ such that
  $\tuple{v,v'} \in E$, $\Lambda(v,v')$ is a transition formula, and
  $[s,s'] \models \Lambda(v,v')$.
\item (Call) There is a transition
  \[\tr{\textit{TS}(P)}{\tuple{v,s}\textit{base}}{\tuple{\textit{entry}(\Lambda(v,v')),s}\tuple{v',s}\textit{base}}\]
  for any stack $\textit{base}$, any activation record $\tuple{v,s}$,
  and any edge $\tuple{v,v'}$ labeled by a call.
\item (Return) There is a transition
\[\tr{\textit{TS}(P)}{\tuple{\textit{exit}(p),s_1}\tuple{v,s_2}\textit{base}}{\tuple{v,s}\textit{base}}\] for any stack $\textit{base}$, any procedure
  $p$, any state $s_1$, and any activation record
  $\tuple{v,s_2}$, where $s$ is the state defined by
  \[ s(x) \defeq \begin{cases}
    s_1(x) & \text{if }x \in \GVar\\
    s_2(x) & \text{if }x \in \LVar
  \end{cases}\ .\]
\end{itemize}

\subsection{Procedure summarization} \label{sec:procedure-summaries}

A \textit{summary assignment} is a function $\summary : P \rightarrow \TF$
mapping each procedure to a transition formula.   Given a summary assignment
$s$, we can define a semantic function $L_\summary : E \rightarrow \TF$ by\\
\[L_\summary(e) \defeq \begin{cases}\summary(\Lambda(e)) & \text{if } e \text{ is a call edge } (\Lambda(e) \in P)\\
  \Lambda(e) & \text{if } e \text{ is a transition edge } (\Lambda(e) \in \TF)
\end{cases}\]

A \textbf{closure operator} on transition formulas is a function $\rho
:\mathbf{TF} \rightarrow \mathbf{TF}$ that is:
\begin{itemize}
\item (Monotone): for all $T_1,T_2$ with $T_1 \models T_2$, we have $\rho(T_1) \models \rho(T_2)$
\item (Extensive): for all $T$, we have $T \models \rho(T)$
\item (Idempotent): for all $T$, we have $\rho(\rho(T)) \equiv \rho(T)$
\end{itemize}
We say that a transition formula is \textbf{closed} under $\rho$ if
$\rho(T) \equiv T$.  We say that $\rho$ satisfies the
\textbf{ascending chain condition} if for every infinite chain $T_1
\models T_2 \models T_3 \dotsi$ of transition formulas that are closed
under $\rho$ eventually stabilizes (there exists some $m$ such that
for all $n \geq m$ we have $T_i \equiv T_m$).

\begin{example}[Closure operator]
  Two simple closure operators that satisfy the ascending chain
  condition are as follows:
  \begin{itemize}
  \item Fix a set of predicates $P$, then define $\rho_P(T) \defeq
    \bigwedge_{p \in P} T \models p$.
  \item Define $\rho_{\textit{aff}}(T) \defeq A\vec{x}' = B\vec{x} +
    \vec{c}$, where $A\vec{x}' = B\vec{x} + \vec{c}$ is a
    representation of the affine hull of $T$.  The affine hull can be computed using
    the algorithm from \cite{VMCAI:RSY2004}.
  \end{itemize}

  Finally, observe that closure operators can be combined.  In our
  implementation we use the closure operator: \[ \rho(T) \defeq
  \rho_P(T) \land \rho_{\textit{aff}}(T) \] where
  \[ P \defeq \{ x \bowtie x' : x \in \Var, x' \in \Var', \bowtie \in \{>,\geq,=,\leq,<\} \} \]
  the set of ordering predicates between primed and unprimed
  variables.
\end{example}

We define an infinite sequence $\summary_0,\summary_1,\dots : P \rightarrow \mathbf{TF}$ of summary assignments where
\begin{align*}
  \summary_0(p) &\defeq \false\\
  \summary_{i+1}(p) &\defeq \rho(\exists \LVar,\LVar'. M(p,\summary_i)) \land \bigwedge_{x \in \LVar} x = x'
\end{align*}
where
\[
M(p,\summary_i) \defeq \mathcal{T}_{\summary_i}\sem{\pathexp{G}{\textit{entry}(p)}{\textit{exit}(p)}}
\]
It follows from the fact that $\rho$ is a closure operator
satisfying the ascending chain condition that there exists some $i$
such that $\summary_i = \summary_{i+1}$; define
$\summary$ to be $\summary_i$ for the least such $i$.

\begin{lemma}
  For any procedure $p$, states $s,s'$, and stack $\textit{st}$ such
  that $\trstar{\textit{TS}(P)}{\tuple{\textit{entry}(p),s}\textit{st}}{\tuple{\textit{exit}(p),s'}\textit{st}}$,
  we have that $[s,s'] \models \summary(p)$.
\end{lemma}

%%% Local Variables:
%%% TeX-master: "paper"
%%% End:

\end{document}